\documentclass[10pt]{article}

\usepackage[numbers,square]{natbib}			% need for \citet*
\usepackage{amsmath}    								% need for subequations
\usepackage{graphicx}   								% need for figures
\usepackage{subfigure}									% need for subfigures
\usepackage{amssymb}    								% gives you \mathbb{} font
\usepackage[mathscr]{eucal}							% gives you \mathscr font
\usepackage{cancel}											% gives you the ability to visibly cross out terms in equations}
\usepackage{comment}
\usepackage{color}
\usepackage{hyperref}
\usepackage[normalem]{ulem} 						% gives you \sout
\usepackage{pstricks}
\usepackage{rotating}
\usepackage{lscape}
\usepackage[paperwidth=8.5in,paperheight=11in,top=1.25in, bottom=1.25in, left=1.00in, right=1.00in]{geometry}
\usepackage{mathtools}									% need for `show only references'
\mathtoolsset{showonlyrefs=true}		% only equations which are labeled AND referenced will be numbered.
\usepackage{amsthm}					% need for theorem-proof environment
\allowdisplaybreaks
\theoremstyle{plain}
\newtheorem{theorem}{Theorem}
\newtheorem{lemma}[theorem]{Lemma}	% [theorem] ==> theorems lemmas, props and corollaries will share a counter
\newtheorem{proposition}[theorem]{Proposition}	

\theoremstyle{definition}
\newtheorem{definition}[theorem]{Definition}
\newtheorem{example}[theorem]{Example}
\newtheorem{remark}[theorem]{Remark}
\newtheorem{assumption}[theorem]{Assumption}

\newcommand{\<}{\langle} 
\renewcommand{\>}{\rangle}
\renewcommand{\(}{\left(}				
\renewcommand{\)}{\right)}
\renewcommand{\[}{\left[}
\renewcommand{\]}{\right]}			

\newcommand\Cb{\mathbb{C}}	
\newcommand\Eb{\mathbb{E}}

\newcommand\Rb{\mathbb{R}}										
\newcommand\Ib{1\hspace{-2.1mm}{1}}

\newcommand\Lc{\mathscr{L}}	
	
\newcommand\Nc{\mathscr{N}}		
\newcommand\Oc{\mathscr{O}}

\newcommand\BS{\mathrm{BS}}
\newcommand\SVI{\mathrm{SVI}}

\newcommand\eps{\varepsilon}

\newcommand\sig{\sigma}

\newcommand\Gam{\Gamma}
\newcommand\lam{\lambda}
\newcommand\del{\delta}

\newcommand\hh{\widehat{h}}

\renewcommand\d{\partial}		

\newcommand\dd{\mathrm{d}}
\newcommand\ii{\mathtt{i}}
\newcommand{\ee}{\mathrm{e}}

\newcommand{\I}{\mathtt{i}}

\newcommand{\E}{\mathrm{e}}

\begin{document}

\title{From characteristic functions to implied volatility expansions}

\author{
Antoine Jacquier
\thanks{Department of Mathematics, Imperial College London, London, United Kingdom.}
\and
Matthew Lorig
\thanks{Department of Applied Mathematics, University of Washington, Seattle, WA, USA.  Work partially supported by NSF grant DMS-0739195.}
}

\date{This version: \today}

\maketitle
%\tableofcontents

\begin{abstract}
For any strictly positive martingale $S = \E^X$ for which $X$ has a characteristic function, 
we provide an expansion for the implied volatility.
This expansion is explicit in the sense that it involves no integrals, but only polynomials in the log strike.
We illustrate the versatility of our expansion by computing the approximate implied volatility smile 
in three well-known martingale models: one finite activity exponential L\'evy model (Merton), 
one infinite activity exponential L\'evy model (Variance Gamma), and one stochastic volatility model (Heston).  
Finally, we illustrate how our expansion can be used to perform a model-free calibration 
of the empirically observed implied volatility surface.
\end{abstract}


%%%%%%%%%%%%%%%%%%%%%%%%%%%%%%%%%%%%%%%%%%%%%%%%%%%%
%
%		SECTION: Introduction
%
%%%%%%%%%%%%%%%%%%%%%%%%%%%%%%%%%%%%%%%%%%%%%%%%%%%%

\section{Introduction}
\label{sec:intro}
While it is rare to find a martingale model for which the transition density is available in closed form 
(the Black-Scholes model being a notable exception), 
there is a veritable zoo of models for which the characteristic function is available explicitly
(exponential L\'evy models  and affine models~\cite{duffiepansingleton} for instance).
The existence of an analytically tractable characteristic function allows for (vanilla) option prices to be computed
quickly using (generalised) Fourier transforms~\cite{lewis2001simple,lipton2002}.
Every model contains unobservable parameters, which are usually calibrated to market data.
This calibration procedure is typically performed using implied volatilities rather than option prices, 
the former being dimensionless.
For a given model, one therefore has to compute (by finite difference, Monte Carlo or numerical integration) 
option prices first and then the corresponding implied volatilities by some root-finding algorithm.
Both steps require sophisticated numerical tools and occasionally somewhat of an artistic touch.
These are computationally expensive and render calibration a long and intensive task.

Over the past decade, many authors have focused on obtaining closed-form approximations 
for both option prices and implied volatilities, partly in order to speed up this calibration process.
Perturbation methods have been used by Lorig and co-authors~\cite{lpp-2} 
(see also~\cite{fouque2003proof,fpss,Fuka2011,sabr})
to obtain such approximations for diffusion-type models.
In extreme regions---where numerical schemes become less efficient---asymptotic expansions 
of densities and of implied volatilities have been obtained in~\cite{BBF2004,DFJV2013,HGLOW2012,LabordereBook,Tankov2010}
in the small-maturity case (both for diffusions and jump models) and in~\cite{JKRM2013} for the large-time behaviour
of affine stochastic volatility models.
Roger Lee~\cite{Lee2004} pioneered the study of the tails of implied volatility, 
and more recent (model-dependent and model-free) results have appeared
in~\cite{BF2009,DFJV2013,FGGS2011,GuliStein2010}.

The goal of this paper is to derive an approximation for the implied volatility in any model whose characteristic function is available in closed form.
This approximation contains no special function and does not require any numerical integration.
It can therefore be used efficiently to accelerate the aforementioned calibration issue.
The methodology follows and extends the previous works~\cite{lorig-jacquier-1,lorigCEV} and is related to some extent 
to the works by Takahashi and Toda~\cite{Toda2013}.
Indeed, by writing the characteristic function as a perturbation around the Black-Scholes characteristic function, 
our expansion has the form of a Black-Scholes price perturbed by some additional quantity (which we shall make precise later), which can then be turned into an expansion for the corresponding implied volatility.

The rest of the paper proceeds as follows:
in Section~\ref{sec:pricing}, we provide a brief review of the characteristic function approach to option pricing
and introduce some notations needed later in the paper.
Section~\ref{sec:iv} contains the main results, namely a series expansion for the implied volatility.
More precisely, we show (Section~\ref{sec:u.expand}) that, whenever the characteristic function is available in closed-form, the European call price can be written as a regular perturbation around the Black-Scholes price.  
A similar result then holds for the implied volatility, 
as detailed in Sections~\ref{sec:sig.expand} and~\ref{sec:simple}.
In Section~\ref{sec:examples} we numerically test our results and provide practical details about this implementation.

%%%%%%%%%%%%%%%%%%%%%%%%%%%%%%%%%%%%%%%%%%%%%%%%%%%%
%
%		SECTION: pricing
%
%%%%%%%%%%%%%%%%%%%%%%%%%%%%%%%%%%%%%%%%%%%%%%%%%%%%

\section{Notations and preliminary results}
\label{sec:pricing}
We consider here a given probability space $\left(\Omega,(\mathcal{F}_t)_{t\geq 0},\mathbb{P}\right)$;
all the processes studied will be $\mathcal{F}$-adapted.
In particular $S=\E^X$ will denote the stock price process, namely a $\mathcal{F}$-adapted martingale under the risk-neutral probability measure $\mathbb{P}$.
The dynamics of $X$ may depend on some auxiliary process $Y\in\Rb^m$ ($m\geq 1$), 
say some stochastic volatility.
The starting point $(X_0,Y_0)=(x,y)$ is assumed to be non-random.
For simplicity and notational convenience, we will assume that $m=1$ and that the risk-free interest rate is zero.  

\subsection{Pricing via Fourier transforms}
Let $h$ be the payoff function of a European call option on $S=\E^{X}$ with strike $\E^{\zeta}$:
$h(z) \equiv (\E^z - \E^\zeta)^+$, and denote $\hh$ its (generalised) Fourier transform
$$
\hh(\lam) := \int_\Rb \E^{- \ii \lam z} h(z) \dd z
 = \frac{-\E^{\zeta-\ii \zeta \lam}}{ \ii \lam + \lam^2 } , 
\qquad\text{for }
\Im(\lam) < - 1.
$$
The results obtained below for option prices remain valid for Put options with payoff $h(z) \equiv (\E^\zeta - \E^z)^+$, 
but we shall chiefly consider European call option prices unless otherwise stated.
For any $t\geq 0$, define the moment explosions
$p^*(t) := \sup\{ p\geq 0 : \Eb_{x} \left(\E^{p X_t}\right) < \infty \}$
and
$q^*(t) := \sup\{ q\geq 0 : \Eb_{x} \left(\E^{-q X_t}\right) < \infty \}$.
Since~$S$ is a martingale, we have $p^*(t)\geq 1$ and $q^*(t)\geq 0$.
We shall further make the stronger assumption:
\begin{assumption}\label{ass:Explosions}
For any $t\geq 0$, $p^*(t)>1$ and $q^*(t)>0$.
\end{assumption}
This assumption holds for most models in practice,
and allows us to write  the value of a call option as
\begin{align}\label{eq:InvFourierPrice}
u(t,x) := \Eb_{x} \, h(X_t)
 = \frac{1}{2\pi} \int_\Rb \hh(\lam) \Eb_{x} \left(\E^{\ii \lam X_t}\right)\dd \lam_r,
\qquad \text{with }\Im(\lambda)\in (-p^*(t), -1), \text{ for all }t\geq 0,
\end{align}
where we write $\lam = \lam_r + \ii \lam_i $ ($\lambda_r, \lambda_i \in\mathbb{R}$) for a complex number.
Of course the function~$u$ also depends on~$y$, the starting point of~$Y$, but we shall omit it in the notations for clarity.
In this paper, we consider models for which the characteristic function 
$\Cb\ni\lambda\mapsto \Eb_{x} \left(\E^{\ii \lam X_t}\right)$  admits the representation
\begin{align}\label{eq:char.form}
\log \Eb_{x}\left(\E^{\ii \lam X_t}\right) = \ii \lam x + \phi(t,\lam),
\end{align}
for some analytic function $\phi:\Rb_+\times \Cb \to \Cb$, satisfying 
$\phi(t,-\ii)=0$ for all $t\geq 0$ (martingale property).
From~\eqref{eq:InvFourierPrice}, this implies that the price of a call option may be written as (see also~\cite{lewis2001simple} or~\cite{lipton2002})
\begin{align}
u(t,x) &= \frac{1}{2\pi} \int_\Rb \hh(\lam) \E^{ \ii \lam x + \phi(t,\lam)}\dd \lam_r.
\end{align}
Several well-known models fit within this class
\begin{align}
\text{L\'evy models}:&&
\phi(t,\lam)
	&=	t \( \ii \mu\lam - \tfrac{1}{2}a^2 \lam^2 + \int_\Rb \nu(dz) ( \E^{\ii \lam z} - 1 - \ii \lam z ) \) , \\
\text{Additive models}:&&
\phi(t,\lam)
	&=	\ii \mu(t)\lam - \tfrac{1}{2}a^2(t) \lam^2 + \int_\Rb \nu(t,dz) ( \E^{\ii \lam z} - 1 - \ii \lam z )  , \\
\text{Affine models}:&&
\phi(t,\lam)
	&=	C(t,\lam) + y D(t,\lam), \label{eq:affine}
\end{align}
where $(\mu,a^2,\nu)$ is a L\'evy triplet, $(\mu(t),a^2(t),\nu(t))$ are the spot characteristics of an additive process, the function $C$ is fully characterised by $\tfrac{\dd}{\dd t} C = D$ and the function $D$ satisfies a Riccati equation.
For precise details on L\'evy and affine processes, 
we refer the interested reader to the monograph by Sato~\cite{sato1999levy} and the groundbreaking paper by Duffie, Filipovi\'c and Schachermayer~\cite{DuffFilipScha}.

%%%%%%%%%%%%%%%%%%%%%%%%%%%%%%%%%%%%%%%%%%%%%%%%%%%%
%%%%%%%%%%%%%%%%%%%%%%%%%%%%%%%%%%%%%%%%%%%%%%%%%%%%
\subsection{Black-Scholes and implied volatility}
\label{subsec:imp.vol}
Option prices are commonly quoted in units of implied volatility (rather than in units of currency)
first because the latter is dimensionless, and second, because the shape and behaviour of the implied volatility
provide more information than option prices.
However, the implied volatility is scarcely available in closed form and has to be computed numerically via inversion of the Black-Scholes formula.
We derive here a closed-form expansion for the implied volatility for models whose characteristic function is of the form~\eqref{eq:char.form}.
We begin our analysis by defining the Black-Scholes price and the implied volatility.
\begin{definition}\label{def:BS}
The \emph{Black-Scholes price} $u^{\BS}:\Rb^+\times\Rb\times\Rb^+ \to \Rb^+$ is given by
\begin{align}
u^{\BS}(t,x,\sig_0) := \frac{1}{2\pi} \int_\Rb \E^{t \phi_0(\lam;\sig_0)} \hh(\lam) \E^{\ii \lam x}\dd \lam_r,
\qquad\text{where}\qquad
\phi_0(\lam;\sig_0) := -\frac{1}{2}\sig_0^2\left(\lam^2 + \ii \lam\right).
\label{eq:uBS}
\end{align}
\end{definition}
\begin{remark}
Note that $\phi_0(\cdot;\sig_0)$ is the L\'evy exponent of a Brownian motion with volatility $\sig_0$ and drift $-\frac{1}{2}\sig_0^2$, 
so that~\eqref{eq:uBS} is the Fourier representation of the usual Black-Scholes price, more typically written as
\begin{align}
u^{\BS}(t,x,\sig_0)	
	&= \E^x \Nc(d_+(x)) - \E^k \Nc(d_-(x)), & 
d_\pm(x) 
	&:=	\frac{1}{\sig_0 \sqrt{t}} \( x - \zeta \pm \frac{1}{2}\sig_0^2 t \), \label{eq:uBS.2}
\end{align}
where $\Nc$ is the cumulative distribution function of a standard normal random variable.
\end{remark}
\begin{definition}\label{def:imp.vol.def}
For any maturity $t$, starting point $x$ and (log) strike $\zeta$, 
the implied volatility is defined as the unique non negative real solution $\sig$ to the equation
$u^{\BS}(t,x,\sig) = u$, 
where $u$ is the (observed or computed) call option price with the same maturity and log strike.
\end{definition}
\begin{remark}
For any $t>0$, the existence and uniqueness of the implied volatility can be deduced using the general arbitrage bounds for call prices 
and the monotonicity of $u^{\BS}$ (see~\cite[Section 2.1, Remark (i)]{fpss}).
\end{remark}
For any $t\geq 0$, $x\in\Rb$, the function $u^{\BS}(t,x,\cdot)$ is analytic on $(0,\infty)$, and hence
for any~$\sig_0>0$ and~$\delta\in\Rb$ such that~$\sig_0 + \del>0$, 
the function $u^{\BS}(t, x, \cdot)$ at the point~$\sig_0 + \del$ is given by its Taylor series:
\begin{align}\label{rmk:analytic}
u^{\BS}(t,x,\sig_0 + \del) = \sum_{n=0}^\infty \frac{\del^n}{n!} \d_\sig^n u^{\BS}(t,x,\sig_0), 
\end{align}
where
$
\d_\sig^n u^{\BS}(t,x,\sig_0) 
= \frac{1}{2\pi} \int_\Rb \left.\(\d_\sig^n \E^{t \phi_0(\lam;\sig)} \)\right|_{\sigma=\sigma_0}
 \hh(\lam) \E^{\ii \lam x}\dd \lam_r.
$
The interchange of the derivative and integral operators is justified by Fubini's theorem.
If one observes the option price~$u$, 
then the following proposition provides a way to compute the corresponding implied volatility.
\begin{proposition}
\label{thm:lagrange}
For any $t> 0$, $x\in\Rb$, let $u:(0,\infty)\to(0,\infty)$ be defined (as a function of $\sig$) by $u^{\BS}(t, x, \sig)=u$,
and let $\sigma_0$ be some strictly positive real number.
Then the following expansion holds:
\begin{align}
\sig 
	&= \sig_0 + \sum_{n=1}^\infty \frac{b_n}{n!} (u - u^{\BS}(t, x, \sig_0))^n, &
b_n 
	&:= \lim\limits_{\sig \to \sig_0} \d_\sig^{n-1} \( \frac{ \sig - \sig_0}{u^{\BS}(t,x,\sig)-u^{\BS}(t,x,\sig_0)} \)^n .
			\label{eq:inverse}
\end{align}
%where
%$b_n := \lim\limits_{\sig \to \sig_0} \d_\sig^{n-1} \( \frac{ \sig - \sig_0}{u^{\BS}(t,x,\sig)-u^{\BS}(t,x,\sig_0)} \)^n$.
\end{proposition}
\begin{proof}
Since the function $u^{\BS}(t,x,\cdot)$ is strictly increasing on~$(0,\infty)$, 
analytic in a neighbourhood of~$\sig_0$ 
and $\d_\sig u^{\BS}(\cdot,\cdot,\sig_0) \neq 0$, 
the proposition follows from Lagrange Inversion Theorem~\cite[Equation 3.6.6]{as}.
\end{proof}
Proposition~\ref{thm:lagrange} shows that, for every fixed $t> 0$, $x\in\Rb$, $\sig_0>0$, 
there exists some radius of convergence $R>0$ (depending on~$(t,x,\zeta)$) such that 
$|u - u^{\BS}(t,x,\sig_0)|<R$ 
implies that $\sig$, defined implicitly through the equation $u^{\BS}(t,x,\sig)=u$, is fully characterised by~\eqref{eq:inverse}.
This result however seems to be only of theoretical interest. 
Once the option value~$u$ is known, computing the implied volatility inverting the Black-Scholes formula is a simple numerical exercise.
Moreover, computing the implied volatility using~\eqref{eq:inverse} is not numerically efficient since
the option price~$u$ requires the computation of a (possibly highly oscillatory) Fourier integral.
One may wish to use~\eqref{eq:inverse} to deduce some properties of the implied volatility, 
but then the proposition would benefit from precise error bounds when truncating the infinite sum.
The rest of the paper focuses on developing a similar expansion, 
without the need for the (potentially computer-intensive) implementation of the value function $u$.

%%%%%%%%%%%%%%%%%%%%%%%%%%%%%%%%%%%%%%%%
%%%%%%%%%%%%%%%%%%%%%%%%%%%%%%%%%%%%%%%%
\section{Implied volatility expansions}
\label{sec:iv}

%%%%%%%%%%%%%%%%%%%%%%%%%%%%%%%%%%%%%%%%
%%%%%%%%%%%%%%%%%%%%%%%%%%%%%%%%%%%%%%%%
\subsection{Call prices as perturbations around Black-Scholes}
\label{sec:u.expand}
For any $\eps\in (0,1]$ and $\sigma_0>0$ define the function $\phi^\eps(\cdot,\cdot;\sigma_0):\Rb_+\times\Cb\to\Cb$ by
$$
\phi^\eps(t,\lam;\sig_0) :=t \, \phi_0(\lam;\sig_0) + \eps \, \phi_1(t,\lam;\sig_0) ,
$$
where $\ee^{t\phi_0}$
is the Black-Scholes characteristic function from Definition~\ref{def:BS} and
\begin{equation}\label{eq:phi1}
\phi_1(t,\lam;\sig_0) := \phi(t,\lam) - t \phi_0(\lam;\sig_0).
\end{equation} 
Recall from Bochner theorem~\cite[Theorem 4.2.2]{Lukacs} that a complex-valued function~$f$
is a characteristic function if and only if it is non-negative definite and $f(0)=1$.
Therefore~$\E^{\phi^\eps}$ is a well-defined characteristic function for any $t\geq 0$, and
we can associate to it a (unique up to indistinguishability) stochastic process $(X_t^{\eps,\sigma_0})_{t\geq 0}$, 
starting at $X_0^{\eps,\sigma_0}=x$, which is a true martingale.
The price $u^\eps$ of a call option written on~$X^{\eps,\sigma_0}$ thus reads
\begin{align}\label{eq:u.eps}
u^\eps(t,x,\sig_0) :=	\frac{1}{2\pi} \int_\Rb \dd \lam_r \, \E^{\phi^\eps(t,\lam;\sig_0)} \hh(\lam) \E^{\ii \lam x}.
\end{align}
Let $\sig^\eps$ denote the implied volatility corresponding to the option price~$u^\eps(t,x,\sig_0)$.
Since $\phi^\eps |_{\eps=1} = \phi$ and $u^\eps |_{\eps=1} = u$, 
the implied volatility corresponding to the option price~$u$ is given by $\sig = \sig^\eps |_{\eps=1}$.  
We now seek an expression for $\sig^\eps$.  
The first step is to show that $u^\eps$ can be written as a power series in $\eps$, 
whose first term corresponds to the Black-Scholes call price with volatility $\sig_0$.  
To this end, we first expand $\E^{\phi^\eps(t,\lam;\sig_0)}$ as
$$
\exp\left(\phi^\eps(t,\lam;\sig_0)\right)
 = \E^{t \phi_0(\lam;\sig_0)} \sum_{n=0}^\infty \frac{1}{n!} \eps^n \phi_1^n(t,\lam;\sig_0),
$$
and deduce a series representation for $u^\eps$ in~\eqref{eq:u.eps}:
\begin{align}\label{eq:u.eps.expand}
u^\eps(t,x,\sig_0) = \sum_{n=0}^\infty \eps^n u_n(t,x,\sig_0) ,
\qquad\text{with}\qquad
u_n(t,x,\sig_0) := \frac{1}{n!} \frac{1}{2 \pi} 
\int_\Rb \dd \lam_r \E^{t \phi_0(\lam;\sig_0)}  \phi_1^n(t,\lam;\sig_0) \hh(\lam) \E^{ \ii \lam x},
\end{align}
for any $n\geq 0$, where the application of Fubini's theorem is justified since 
$\int_\Rb  \left|\E^{t \phi^\eps(\lam)} \hh(\lam) \E^{\ii \lam x}\right|\dd \lam_r$ is finite.
Note in particular that $u_0 \equiv u^{\BS}$.

%%%%%%%%%%%%%%%%%%%%%%%%%%%%%%%%%%%%%%%%%%%%%%%%%%%
\subsection{Series expansion for implied volatility}
\label{sec:sig.expand}
From \eqref{eq:u.eps.expand}, it is clear that $u^\eps$ is an analytic function of $\eps$ (we have explicitly provided its power series representation).  
Since the composition of two analytic functions is also analytic~\cite[Section 24, p. 74]{brown1996complex}, 
the expansion~\eqref{rmk:analytic} implies that $\sig^\eps = [u^{\BS}]^{-1}(u^\eps)$ is an analytic function and therefore has a power series expansion in $\eps$, which we write 
$\sig^\eps := \sig_0 + \del^\eps$, where
$\del^\eps = \sum_{k\geq 1}\eps^k \sig_k$.
The following proposition provides an expansion formula for the coefficients~$\sigma_k$.
\begin{proposition}
Fix $\sigma_0>0$, $k\geq 1$, and let $R$ denote the radius of convergence of the expansion~\eqref{eq:inverse}.
If $\left|u(t,x)-u^{\BS}(t,x,\sig_0)\right|<R$ for all $(t,x)\in\mathbb{R}_+\times \mathbb{R}$, then the following expansion holds:
\begin{align}\label{eq:sig.k}
\sig_k  = \frac{1}{\d_\sig u^{\BS}(t,x,\sig_0)}
\( u_k - \sum_{n=2}^\infty \frac{1}{n!}\( \sum_{j_1+\cdots+j_n=k} \prod_{i=1}^n \sig_{j_i} \)  \d_\sig^n u^{\BS}(t,x,\sig_0)\). 
\end{align}
\end{proposition}
The right-hand side only involves $\sig_j$ for $j \leq k-1$, so that the sequence can be determined recursively.
\begin{proof}
Let us fix some $t\geq 0$ and $x\in\Rb$.
Taylor expanding $u^{\BS}(t,x,\sig^\eps)$ around the point $\sig_0$ we obtain
\begin{align*}
u^{\BS}(t,x,\sig^\eps)
	&=	u^{\BS}(t,x,\sig_0 + \del^\eps) 
	=	\sum_{n=0}^\infty \frac{1}{n!}(\del^\eps \d_\sig )^n u^{\BS}(t,x,\sig_0) \\
	&=	u^{\BS}(t,x,\sig_0) +	\sum_{n=1}^\infty \frac{1}{n!} \( \sum_{k=1}^\infty \eps^k \sig_k \)^n \d_\sig^n u^{\BS}(t,x,\sig_0) \\
	&=	u^{\BS}(t,x,\sig_0) + \sum_{n=1}^\infty \frac{1}{n!}  
		\[ \sum_{k=1}^\infty \( \sum_{j_1+\cdots+j_n=k} \prod_{i=1}^n \sig_{j_i} \) \eps^k \] \d_\sig^n u^{\BS}(t,x,\sig_0) \\
	&=	u^{\BS}(t,x,\sig_0) +	\sum_{k=1}^\infty \eps^k 
		\[ \sum_{n=1}^\infty \frac{1}{n!} \( \sum_{j_1+\cdots+j_n=k} \prod_{i=1}^n \sig_{j_i} \) \d_\sig^n \] u^{\BS}(t,x,\sig_0) \\
	&=	u^{\BS}(t,x,\sig_0) +	\sum_{k=1}^\infty \eps^k 
		\[ \sig_k \d_\sig + \sum_{n=2}^\infty \frac{1}{n!}\( \sum_{j_1+\cdots+j_n=k} \prod_{i=1}^n \sig_{j_i} \)  \d_\sig^n \] u^{\BS}(t,x,\sig_0).
\end{align*}
In order to recover the implied volatility from Definition~\ref{def:imp.vol.def}, 
we need to equate the Black-Scholes call price above and 
the option value~$u^\eps$ in~\eqref{eq:u.eps.expand}, and collect terms of identical powers of $\eps$:
\begin{align}
\Oc(1):&&
u_0(t,x,\sig_0) &= u^{\BS}(t,x,\sig_0) , \\
\Oc(\eps^k):&&
u_k(t,x,\sig_0)
	&=		\sig_k \d_\sig u^{\BS}(t,x,\sig_0)
				+ \sum_{n=2}^\infty \frac{1}{n!}\( \sum_{j_1+\cdots+j_n=k} \prod_{i=1}^n \sig_{j_i} \)  \d_\sig^n	u^{\BS}(t,x,\sig_0) , &
k & \geq 1 .
\end{align}
Solving the above equations for the sequence $(\sig_k)_{k\geq 0}$, we find
$\sigma_0=\sigma_0$ at the zeroth order and for any $k\geq 1$, the $\Oc(\eps^k)$ order is given by~\eqref{eq:sig.k}.
\end{proof}
\begin{remark}
%Explicitly, up to $\Oc(\eps^5)$ we have
Explicitly, up to $\Oc(\eps^3)$ we have
\begin{align}
\begin{aligned}
\Oc(\eps):&&
\sig_1
	&= 	\frac{1}{\d_\sig u^\BS} u_1, \\
\Oc(\eps^2):&&
\sig_2
	&= 	\frac{1}{\d_\sig u^\BS} \Big( u_2 - \tfrac{1}{2} \sig_1^2 \d_\sig^2 u^\BS \Big), \\
\Oc(\eps^3):&&
\sig_3
	&= 	\frac{1}{\d_\sig u^\BS} \Big( u_3 - \(\sig_2 \sig_1 \d_\sig^2 + \tfrac{1}{3!}\sig_1^3 \d_\sig^3 \) u^\BS \Big), %\\
%\Oc(\eps^4):&&
%\sig_4
	%&= \frac{1}{\d_\sig u^\BS} \Big( u_4 - \left(\tfrac{1}{2}\left(2 \sig_3 \sig_1 + \sig_2^2\right) \d_\sig^2 
%+ \tfrac{1}{2} \sig_2 \sig_1^2 \d_\sig^3 + \tfrac{1}{4!} \sig_1^4 \d_\sig^4 \right) u^\BS \Big) , \\
%\Oc(\eps^5):&&
%\sig_5
	%&=	\frac{1}{\d_\sig u^\BS} \bigg( u_5
			%- \Big(\left(\sig_1 \sig_4 + \sig_2 \sig_3\right) \d_{\sig}^2
			%+ \tfrac{1}{2} \left(\sig_1^2 \sig_3 + \sig_1 \sig_2^2 \right) \d_{\sig}^3
			%+ \tfrac{1}{3} \sig_1^3 \sig_2 \d_{\sig}^4
			%+ \tfrac{1}{5!} \sig_1^5 \d_{\sig}^5 \Big) u^\BS
			%\bigg),
\end{aligned}
\label{eq:sigmas}
\end{align}
where all the functions $u_1,\ldots,u_5$ and $u^\BS$ are evaluated at $(t,x,\sigma_0)$.
\end{remark}
\begin{remark}
Having served its purpose, we now dial $\eps$ up to one.  
The implied volatility is then given by
$\sig = \sum_{k\geq 0}\sig_k$, 
where $\sig_0$ is a fixed positive constant and 
where the sequence~$(\sig_k)_{k\geq 1}$ is given by~\eqref{eq:sig.k}.
\end{remark}

%%%%%%%%%%%%%%%%%%%%%%%%%%%%%%%%%%%%%%%%%%%%%%%%5
\subsection{Simplification of the expressions for $\sig_k$}
\label{sec:simple}
The expression for the coefficients $\sig_k$ ($k\geq 1$) in \eqref{eq:sig.k} is not straightforward to apply; 
one needs to compute first the Fourier integrals $u_j$  ($j\leq k$) via~\eqref{eq:u.eps.expand}, 
then all the terms of the form $\d_\sig^j u^{\BS}$ ($j\leq k$). 
We provide now a more explicit approximation---without integrals or special functions---for $\sig_k$. 
The key to this simplification is that all the terms $\partial^i_\sigma u^{\BS}$ and $u_i$ ($i\in\mathbb{N}$)
in~\eqref{eq:sig.k} can actually be expressed in terms of derivatives of $u^{\BS}$ with respect to $x$,
the starting point of the log stock price process.
Indeed, the classical Black-Scholes relation between the Delta, the Gamma and the Vega for call options,
$\left.\d_\sig u^{\BS}(t,x,\sig)\right|_{\sigma=\sigma_0} = t \sig_0 ( \d_x^2 - \d_x ) u^{\BS}(t,x,\sig_0) $, 
implies that the derivative~$\d_\sig^k u^{\BS}$ can be expressed as a sum of terms of the form 
$a_{k_i} \d_x^{k_i} (\d_x^2 - \d_x) u^{\BS}$.
We shall also use the equality
$p(\lam) \E^{\ii \lam x} = p(- \ii \d_x) \E^{\ii \lam x}$, which holds for any polynomial $p$
(and actually for any analytic function---simply take $p$ to be its power series).  
We first start with the following theorem, which provides an approximation for the coefficients~$u_n$ in~\eqref{eq:u.eps.expand} 
as a differential operator acting on $u^{\BS}$.
\begin{theorem}\label{thm:SimpleUn}
Fix some $t\geq 0$ and $\sigma_0>0$.
If the power series
$\phi_1(t,\lam;\sig_0) = \sum_{k\geq 1} a_k(t;\sig_0) (\ii \lam)^k$
holds in a complex neighbourhood of the origin, 
then for any integer $m\geq 2$, $u_n$ defined in~\eqref{eq:u.eps.expand} can be written as
$u_n(t,x,\sigma_0)  = u_n^{(m)}(t,x,\sigma_0) + \eps_n^{(m)}(t,x,\sigma_0)$,
where 
\begin{align}
u_n^{(m)}(t,x,\sig_0)
	&=	\frac{1}{n!}\( \sum_{k=2}^m a_k(t,\sig_0) (\d_x^k - \d_x )\)^n u^{\BS}(t,x,\sig_0),
\label{eq:u.n.m}
\end{align}
and where $\varepsilon_n^{(m)}$ only contains derivatives (with respect to~$x$) of $u^{\BS}$ of order higher than~$n$.
\end{theorem}
\begin{remark}
Note that the power series for $\phi_1$ in the theorem starts at $k=1$, which
follows from the fact that the process $\exp(X^{\eps,\sig_0})$ is conservative.
This expansion holds as soon as 
all the moments of $X_t$ exist and 
$\lim_{k\uparrow \infty}|\lambda|^k\Eb\left(|X_t|^k\right)/k! = 0$
for $|\lambda|$ small enough, which is valid under Assumption~\ref{ass:Explosions}.
\end{remark}

\begin{proof}
Assume that the power series for $\phi_1(t,\cdot;\sig_0)$ holds around the origin, 
where the coefficients read
\begin{align}\label{eq:Coefak}
a_k(t;\sigma_0) = \frac{(-\ii)^k}{k!}\left.\partial_\lambda ^k \phi_1 (t,\lambda;\sigma_0)\right|_{\lambda=0}.
\end{align}
The martingale condition implies
$\phi_1(t,-\ii;\sig_0) = \sum_{k\geq 1} a_k(t;\sig_0) = 0$, and hence
\begin{align}
\phi_1(t,\lam;\sig_0)  & = \sum_{k=1}^\infty a_k(t;\sig_0) (\ii \lam)^k
  = \ii \lam a_1(t;\sig_0) + \sum_{k=2}^\infty a_k(t;\sig_0)\left[ (\ii \lam)^k - \ii \lam \right] + \ii\lambda\sum_{k=2}^\infty a_k(t;\sig_0)\\
	&=	\sum_{k=2}^\infty a_k(t;\sig_0) \left[ (\ii \lam)^k - \ii \lam \right]. \label{eq:decomposition}
\end{align}
Let now $\phi_1^{(m)}(t,\cdot;\sigma_0):\Cb\to\Cb$ be the truncation of the series~\eqref{eq:decomposition} at the $m$-th order, i.e.
\begin{align}\label{eq:defPhi1m}
\phi_1^{(m)}(t,\lam;\sig_0) := \sum_{k=2}^m a_k(t;\sig_0) \left[ (\ii \lam)^k - \ii \lam \right],
\end{align}
and define the operator $\delta$ acting on $\phi_1^{(m)}$ by
\begin{align}
\del\phi_1^{(m)}(t,\lam;\sig_0) 
	&:=	\phi_1(t,\lam;\sig_0) - \phi_1^{(m)}(t,\lam;\sig_0) 
	= 	\frac{\lam^{m+1}}{2 \pi \ii} \int_{\d\Gam} \frac{\phi_1(t,z;\sig_0) \dd z}{z^{m+1}(z-\lam)}
 + \ii \lam \sum_{k=1}^{m} a_k,
\end{align}
where $\Gam$ is a closed set within the radius of convergence of $\phi_1$, 
and the integral is nothing else than the remainder of the series expansion around the point $\lambda=0$.
Hence for any $n\geq 1$, $u_n$ in~\eqref{eq:u.eps.expand} can be written as
\begin{align}
u_n(t,x,\sig_0)
	&=	\frac{1}{n!} \frac{1}{2\pi} \int_\Rb \dd \lam_r \, 
			\E^{t \phi_0(t,\lam,\sig_0)} \hh(\lam) \( \phi_1(t,\lam;\sig_0) \)^n \, \E^{ \ii \lam x} \\
	&=	\frac{1}{n!} \frac{1}{2\pi} \int_\Rb \dd \lam_r \, 
			\E^{t \phi_0(\lam;\sig_0)} \hh(\lam) 
			\( \phi_1^{(m)}(t,\lam;\sig_0) + \del \phi_1^{(m)}(t,\lam;\sig_0) \)^n \, \E^{ \ii \lam x} \label{eq:un.int} \\
	&=	\frac{1}{n!} \( \phi_1^{(m)}(t,-\ii \d_x;\sig_0) \)^n \frac{1}{2\pi} \int_\Rb \dd \lam_r \, 
			\E^{t \phi_0(\lam;\sig_0)} \hh(\lam) \E^{ \ii \lam x} + \eps_n^{(m)}(t,x,\sig_0) \\
	&=	\frac{1}{n!} \( \phi_1^{(m)}(t,-\ii \d_x;\sig_0) \)^n u^{\BS}(t,x,\sig_0) + \eps_n^{(m)}(\sig_0)
	= 	u_n^{(m)}(t,x,\sig_0) + \eps_n^{(m)}(t,x,\sig_0),
\end{align}
where
\begin{align}
u_n^{(m)}(t,x,\sig_0)
	&:=	\frac{1}{n!} \( \phi_1^{(m)}(t,-\ii \d_x;\sig_0) \)^n u^{\BS}(t,x,\sig_0) , \\
\eps_n^{(m)}(t,x,\sig_0)
	&:=	\frac{1}{n!} \sum_{k=1}^n \binom{n}{k} \(\phi_1^{(m)}(t,-\ii \d_x;\sig_0)\)^{n-k} \(\del\phi_1^{(m)}(t,-\ii \d_x;\sig_0)\)^k u^{\BS}(t,x,\sig_0) .
\end{align}
From the decomposition~\eqref{eq:defPhi1m}, we can write 
$\phi_1^{(m)}(t,-\ii \d_x;\sig_0) = \sum_{k=2}^m a_k(t;\sig_0) (\d_x^k - \d_x)$,
where the coefficients~$a_k$ are defined in~\eqref{eq:Coefak}.
We can now compute
\begin{align}
u_n^{(m)}(t,x,\sig_0) & = \frac{1}{n!} \( \phi_1^{(m)}(t,-\ii \d_x;\sig_0) \)^n u^{\BS}(t,x,\sig_0)
 = \frac{1}{n!} \left(\sum_{k=2}^m a_k(t;\sig_0) (\d_x^k - \d_x)\right)^n u^{\BS}(t,x,\sig_0) ,
 %& = \frac{1}{n!} \sum_{\substack{2\leq k_2,\ldots, k_m\leq k \\ k_2+\ldots+k_m=n}}
%\frac{n!}{k_2! \ldots k_m!}
%\prod_{2\leq j \leq m}a_{j}^{k_j}(t;\sigma_0) \left(\d_x^j - \d_x\right)^{k_j}u^{\BS}(t,x,\sig_0). 
\label{eq:un.m}
\end{align}
which is precisely the expression given in \eqref{eq:u.n.m}.
%\red{Defining $\gam_n^{(m)}$ as in the theorem, the result follows.}
Regarding $\varepsilon_n^{(m)}$, since for any $z \in \Gam$, there exists $M>0$ such that $\left|\phi_1(t,z;\sig_0)/(z-\lam)\right| \leq M$, we have
$\left|\del\phi_1^{(m)}(t,\lambda;\sig_0) \right| \leq M \left(|\lam|/R \right)^{m+1} + |\lambda| \left|\sum_{k=1}^{m} a_k\right|$,
where $R$ denotes the radius of convergence of~$\phi_1$.
The sequence $(a_k)_{k\geq 1}$ is (eventually) decreasing and the sum tends to zero as $m$ tends to infinity, 
so that the sum can be made arbitrarily small.
We then obtain
\begin{align}
\left|\eps_n^{(m)}(t,x,\sig_0)\right| & 
\leq \frac{1}{n!} \sum_{k=1}^n \binom{n}{k} \left|\phi_1^{(m)}(t,-\ii \d_x;\sig_0)\right|^{n-k} 
\left(M \left(|\lam|/R \right)^{m+1} + |\lambda| \left|\sum_{k=1}^{m} a_k\right|\right)^k u^{\BS}(t,x,\sig_0)\\
& \leq \frac{1}{n!} \sum_{k=1}^n \binom{n}{k} \left|\sum_{j=2}^m a_j(t;\sig_0) \left((\ii \lam)^j - \ii \lam \right)\right|^{n-k} 
\left(M \left(|\lam|/R \right)^{m+1} + |\lambda| \left|\sum_{k=1}^{m} a_k\right|\right)^k u^{\BS}(t,x,\sig_0).
\end{align}
One can then readily check that the sum behaves as $\Oc\left(\lambda^n\right)$ as $\lambda$ tends to zero.
Therefore, $\varepsilon_n^{(m)}$ contains derivatives (with respect to~$x$) of $u^{\BS}$ of at least order $n$.
\end{proof}
\noindent
Expression \eqref{eq:u.n.m} motivates the following definition:
\begin{definition}
\label{def:sig.nm}
For any integers $n \geq 0$, $m \geq 2$ and for fixed $\sig_0 > 0$ we define the \emph{$(n,m)$-th order approximation of the implied volatility} as
\begin{equation}\label{eq:sig.n.m}
\sig^{(n,m)} := \sig_0 + \sum_{k=1}^n \sig_k^{(m)}, 
\end{equation}
where, for any $k=1,\ldots,n$, $\sig_k^{(m)}$ is defined as
\begin{align}
\sig_k^{(m)}  
	&:= \frac{1}{\d_\sig u^{\BS}(t,x,\sig_0)} \( u_k^{(m)} - \sum_{n=2}^\infty \frac{1}{n!}
			\( \sum_{j_1+\cdots+j_n=k} \prod_{i=1}^n \sig_{j_i}^{(m)} \)
			\d_\sig^n u^{\BS}(t,x,\sig_0)\) . \label{eq:sig.k.m}
\end{align}
\end{definition}
Note that~$\sig_k^{(m)}$ is obtained from $\sig_k$ by replacing $u_k$ in~\eqref{eq:sig.k}
by its $m$th order approximation~$u_k^{(m)}$.
The following theorem, proved in Appendix~\ref{sec:proof}, is the main result of this paper 
and provides and explicit expression (not involving the derivatives of Black-Scholes) 
for the $(n,m)$-th order approximation of the implied volatility.
\begin{theorem}
\label{thm:main}
Fix $(t,x,\zeta,\sig_0)\in (0,\infty)\times\mathbb{R}\times\mathbb{R}\times(0,\infty)$, and define
\begin{equation}\label{eq:y0}
y_0 := \frac{1}{\sig_0\sqrt{2t}}\(x-\zeta-\frac{1}{2}\sig_0^2 t\).
\end{equation}
For $\phi(t,\lam)$, $\phi_0(\lam,\sig_0)$, $\phi_1(t,\lam;\sig_0)$ and $a_k(t;\sigma_0)$ defined respectively
in~\eqref{eq:char.form},~\eqref{eq:uBS},~\eqref{eq:phi1} and~\eqref{eq:Coefak}, 
the $(n,m)$-th approximation~\eqref{eq:sig.n.m} holds, 
with $\sig_k^{(m)}=U_k^{(m)}-\Sigma_k^{(m)}$ and
\begin{align}
U_n^{(m)}
	& :=	\frac{1}{n! t \sig_0}  \sum_{k=2 n}^{n m} 
			\sum_{\substack{j_1+\cdots+j_n = k \\ 2 \leq j_1,\cdots,j_n \leq m}} 
			\( \prod_{i=1}^n a_{j_i} \) \sum_{k_1=2}^{j_1} \cdots \sum_{k_n=2}^{j_n} 
			\binom{n-1}{m}(-1)^{n-1-m} 
			\\ & \qquad \times
			\(-\frac{1}{\sig_0\sqrt{2t}}\)^{-n-1+m+\sum_{j=1}^n {k_j}}
			H_{-n-1+m+\sum_{j=1}^n {k_j}}(y_0) , \\
\Sigma_k^{(m)}
	& :=	\sum_{n=2}^\infty \frac{1}{n!}
			\( \sum_{j_1 + \cdots + j_n = k} \prod_{i=1}^n \sig_{j_i}^{(m)} \)
			%\\ & \qquad
			\sum_{q=0}^{\left\lfloor n/2 \right\rfloor}\sum_{p=0}^{n-q-1}
			\binom{n-q-1}{p}c_{n,n-2q}\sig_0^{-(q+p)} t^{n-q-1} 
			\left(\sqrt{2 t}\right)^{1-p-n+q} H_{p+n-q-1}(y_0).
\end{align}
Here, $H_n(y) \equiv (-1)^n \E^{y^2} \d_y^n \E^{-y^2}$ is the $n$-th Hermite polynomial,
the coefficients $(c_{n,n-2k})$ are defined recursively by
$c_{n,n}=1$ and $c_{n,n-2q} =(n-2q+1) c_{n-1,n-2q+1} + c_{n-1,n-2q-1}$,
for any integer $q \in \{ 1, 2, \cdots , \left\lfloor n/2 \right\rfloor \}$.
\end{theorem}
\begin{example}
\label{ex:computation}
To illustrate how the above theorem works in practice, we compute $\sig_2^{(3)}$ explicitly.
Fix $(t,x,\zeta,\sig_0)$ and write $a_k = a_k(t;\sig_0)$.
Equation~\eqref{eq:u.n.m} implies
\begin{align}
u_1^{(3)}(t,x,\sig_0)
	&%=	\sum_{k=2}^3 a_k (\d_x^k - \d_x ) u^\BS
	=	\Big\{a_3 \d_x + (a_2 + a_3) \Big\} (\d_x^2 - \d_x )  u^\BS(t,x,\sig_0),\\
u_2^{(3)}(t,x,\sig_0)
	&%=	\frac{1}{2!}\( \sum_{k=2}^3 a_k (\d_x^k - \d_x ) \)^2 u^\BS
	=	\frac{1}{2} \Big\{a_3^2 \d_x^4 + (2 a_2 a_3 + a_3^2)\d_x^3 + (a_2^2-a_3^2)\d_x^2 + (-a_2^2-a_3^2-2 a_2 a_3)\d_x \Big\} (\d_x^2 - \d_x )  u^\BS(t,x,\sig_0).
\end{align}
Next, using Proposition \ref{prop2} we have
\begin{align}
\frac{u_1^{(3)}}{\d_\sig u^\BS}(t,x,\sigma_0)
	&=	\displaystyle \frac{1}{t \sig_0} \left\{
			a_3 \( \tfrac{-1}{\sig_0\sqrt{2t}}\) H_1(y_0) + (a_2+a_3)
			\right\}, \label{eq:ex1} \\
\frac{u_2^{(3)}}{\d_\sig u^\BS}(t,x,\sigma_0)
	&=	\displaystyle \frac{1}{t \sig_0} \left\{ 
			a_3^2 \( \tfrac{-1}{\sig_0\sqrt{2t}}\)^4 H_4(y_0) + (2 a_2 a_3 + a_3^2)\( \tfrac{-1}{\sig_0\sqrt{2t}}\)^3 H_3(y_0)\right.
			\\ &\qquad \qquad
			\displaystyle \left.+ (a_2^2-a_3^2) \( \tfrac{-1}{\sig_0\sqrt{2t}}\)^2 H_2(y_0) + (-a_2^2-a_3^2-2 a_2 a_3) \( \tfrac{-1}{\sig_0\sqrt{2t}}\) H_1(y_0)
			\right\}, \label{eq:ex2}
\end{align}
with $y_0$ defined in~\eqref{eq:y0}.  
From Proposition~\ref{prop1} we then have
$\d_\sig^2 u^\BS =	\( \sig_0^2 \Lc^2 - \Lc \) u^\BS$, where 
$\Lc = t (\d_x^2 - \d_x)$.
Therefore, recalling that $\d_\sig u^\BS = \sig_0 \Lc u^\BS$ we obtain
\begin{align}
\frac{\d_\sig^2 u^{\BS}}{\d_\sig u^{\BS}}(t,x,\sigma_0)
	&=	\frac{1}{t\sig_0} \( \frac{t^2\sig_0^2\d_x^2(\d_x^2-\d_x)u^\BS}{(\d_x^2-\d_x)u^\BS} 
			- \frac{t^2\sig_0^2\d_x(\d_x^2-\d_x)u^\BS}{(\d_x^2-\d_x)u^\BS} 
			+ \frac{t(\d_x^2-\d_x)u^\BS}{(\d_x^2-\d_x)u^\BS}  \)(t,x,\sigma_0)  \\
	&=	\frac{1}{t\sig_0} \(
			t^2\sig_0^2 \( \tfrac{-1}{\sig_0\sqrt{2t}}\)^2 H_2(y_0)
			- t^2\sig_0^2 \( \tfrac{-1}{\sig_0\sqrt{2t}}\) H_1(y_0)
			+ t \)
	=	\frac{(\zeta-x)^2}{t \sig_0^3}-\frac{t \sig_0 }{4} . \label{eq:ex3}
\end{align}
Lastly, from \eqref{eq:sig.n.m}-\eqref{eq:sig.k.m} we have
\begin{align}
\sig^{(2,3)}
	&=	\sig_0 + \sig_1^{(3)} + \sig_2^{(3)} , &
\sig_1^{(3)}
	&= 	\frac{u_1^{(3)}}{\d_\sig u^\BS} , &
\sig_2^{(3)}
	&= 	\frac{u_2^{(3)}}{\d_\sig u^\BS} 
			- \frac{1}{2} \(\sig_1^{(3)}\)^2 \frac{\d_\sig^2 u^\BS}{\d_\sig u^\BS}  . \label{eq:ex4}
\end{align}
The explicit expression for $\sig^{(2,3)}$ can be obtained by inserting \eqref{eq:ex1}, \eqref{eq:ex2} and \eqref{eq:ex3} into \eqref{eq:ex4}.
\end{example}

%%%%%%%%%%%%%%%%%%%%%%%%%%%%%%%%%%%%%%%%%%%%%%%%%%%%
%
%		SECTION: Numerical Examples
%
%%%%%%%%%%%%%%%%%%%%%%%%%%%%%%%%%%%%%%%%%%%%%%%%%%%%

\section{Numerical implementation: discussions and examples}
\label{sec:examples}
We now focus on the practical implementation of the results above, namely 
%Theorem~\ref{prop.sig.k}.
Theorem~\ref{thm:main}.
%Section~\ref{sec:optimal} discusses the (arbitrary) choice of the level $\sigma_0$ and 
Section~\ref{sec:SVI} proposes a smoothing procedure to further enhance the applicability of our methodology.
In Sections~\ref{sec:Levy} and~\ref{sec:Heston}, we implement our implied volatility expansion in two exponential L\'evy models (Merton and Variance Gamma) and one stochastic volatility model (Heston).

%%%%%%%%%%%%%%%%%%%%%%%%%%%%%%%%%%%%%%%%%%%%%%%%
%%%%%%%%%%%%%%%%%%%%%%%%%%%%%%%%%%%%%%%%%%%%%%%%
\subsection{Smoothing with the SVI parameterisation}
\label{sec:SVI}
Option data is often noisy and limited by the number of strikes at which options are liquidly traded.  
In~\cite{gatheral2004parsimonious}, Jim Gatheral introduces the following \emph{Stochastic Volatility Inspired} 
(SVI) parameterisation:
\begin{align}
\sig_t^{\SVI}(\zeta) &= \left\{ \frac{a}{t} + \frac{b}{t} \( \rho (\zeta-x-m) + \sqrt{(\zeta-x-m)^2 + \xi^2} \)\right\}^{1/2},
 \label{eq:SVI}
\end{align}
for any maturity $t>0$, where $a, b\geq 0$, $\xi>0$, $m\in\mathbb{R}$, $\rho\in [-1,1]$.
By fitting the SVI parameterisation to noisy option data, one is able to create a smooth implied volatility smile, which then can be used to interpolate implied volatility between strikes and extrapolate implied volatility to strikes which are not traded.
The density $p^\sig(t,x)$ corresponding to a given implied volatility parameterisation~$\sig(t,\zeta)$ can be computed via the Breeden-Litzenberger formula~\cite{breeden}:
$p^\sig(t,\zeta) = \d_K^2 u^\BS(t,x,\sig(t,\log K);\log K) \big|_{K=\E^\zeta}$.
An implied volatility smile $\zeta \mapsto \sig(t,\zeta)$ is said to be free of butterfly arbitrage 
if the corresponding density is non-negative: $p(t,\cdot)\geq 0$.  
Let $p_t^\SVI(\zeta)$ be the implied volatility smile corresponding to a given SVI parameterisation~\eqref{eq:SVI}.
In general, SVI parameterisation~\eqref{eq:SVI} is not guaranteed to be free of butterfly arbitrage.  
However, for a given set of SVI parameters $(a,b,\rho,m,\xi)$, 
one can easily verify that the corresponding density is non-negative, and therefore free of butterfly arbitrage.  
This and recent arbitrage-free SVI parameterisations have recently been studied in~\cite{gatheral2012arbitrage}
and~\cite{JacquierMartini}, and we refer the interested reader to these papers for more details.
As we shall see in the examples considered in Section \ref{sec:examples}, 
for finite $(n,m)$, the approximate implied volatility $\sig^{(n,m)}$ 
derived in Section~\ref{subsec:imp.vol} has a tendency to oscillate around the true implied volatility 
(see Figures~\ref{fig:Merton},~\ref{fig:VG} and~\ref{fig:Heston}).  
Taking $\sig^{(n,m)}$ to be the true implied volatility could lead to arbitrage opportunities.  
In order to prevent this, we propose to smooth the implied volatility approximation $\sig^{(n,m)}$ 
by fitting the SVI parameterisation to it. 
That is, given a model for the underlying $X$ and a time to maturity $t$, we first compute the approximate implied volatility $\sig^{(n,m)}$ as a function of $\log$-strike $\zeta$,
and then fit an arbitrage-free SVI parameterisation $\sig_t^{\SVI}$ to $\sig^{(n,m)}$ over some range of strikes, usually chosen to be a symmetric interval around $\zeta=x$.

%%%%%%%%%%%%%%%%%%%%%%%%%%%%%%%%%%%%%%%%%%%%%%%%%%%%
% 		Example: Exponential Levy
%%%%%%%%%%%%%%%%%%%%%%%%%%%%%%%%%%%%%%%%%%%%%%%%%%%%

\subsection{Exponential L\'evy models}\label{sec:Levy}
Suppose that $X$ is a L\'evy process with L\'evy triplet $(\mu,a^2,\nu)$.
Then its characteristic function reads
$$
\phi(t,\lam) = t\(\ii \mu\lam + \frac{1}{2} a^2 (\ii \lam)^2 
 + \int_\Rb \nu(\dd z) (\E^{\ii \lam z} - 1 - \Ib_{\{|z| < 1 \}} \ii \lam z) \),
$$
where the drift $\mu$ is constrained by the martingale condition $\phi(t,-\ii) = 0$:
$\mu = - \frac{1}{2} a^2  - \int_\Rb \nu(\dd z) (\E^z - 1 - \Ib_{ \{|z| < 1 \}} z)$.
From the expansion $(\E^{\ii \lam z} - 1 ) = \sum_{n\geq 1} \frac{1}{n!} (\ii \lam z)^n$ 
we can write 
$$
\phi_1(t,\lam;\sig_0) 
\equiv	\phi(t,\lam) - t \, \phi_0(\lam;\sig_0)
\equiv t \( \mu + I_1 + \tfrac{1}{2} \sig_0^2 \) \ii \lam +  \frac{1}{2} t ( a^2 - \sig_0^2 ) (\ii \lam)^2 + t \sum_{n=2}^\infty \frac{1}{n!} I_n (\ii \lam )^n, 
$$
with
$I_1 := \int_{|z| \geq 1 } \nu(\dd z) z$ and 
$I_n  := \int_\Rb \nu(\dd z) z^n$,
for any $n\geq 2$.
The existence of~$I_n$ is equivalent to the finiteness of the $n$th moment of $X$
by~\cite[Theorem 25.3]{sato1999levy}, which is clearly satisfied under Assumption~\ref{ass:Explosions}.
Hence, the coefficients $a_n(t;\sig_0)$ in~\eqref{eq:Coefak} are given by
\begin{align}
a_2(t;\sig_0)
	&=	\frac{t}{2} (a^2 - \sig_0^2  + I_2 ) , &
a_n(t)
	&=	\frac{t}{n!}I_n , \quad n \geq 3 . \label{eq:Levy.a}
\end{align}
We examine two exponential L\'evy models in detail---the Merton model~\cite{merton1976option} 
and the Variance Gamma model~\cite{madan1998variance}---whose L\'evy measures are given by:
\begin{align}
\text{Merton}:&&
\nu(\dd z)
	&= \frac{\alpha }{\sqrt{2 \pi s^2}}\exp \( \frac{-(z-m)^2}{2s^2} \) \dd z, \\
\text{Variance Gamma}:&&
\nu(\dd z)
	&= \alpha \( \frac{\E^{G z}}{-z} \Ib_{\{ z < 0\} } +  \frac{\E^{-M z}}{z} \Ib_{\{ z > 0\} }\) \dd z,
\end{align}
where $\alpha,s,G,M>0$, $m\in\mathbb{R}$.
The Merton model is a finite-activity L\'evy process ($\nu(\Rb)<\infty$), 
whereas the Variance Gamma model has infinite activity ($\nu(\Rb)=\infty$).  
For infinite activity L\'evy processes, one typically takes the diffusion component to be zero, namely $a=0$.  
We now examine the accuracy of the implied volatility expansion above in these models:
the Merton model in Figure~\ref{fig:Merton} and the Variance Gamma in Figure~\ref{fig:VG}.
For each of these two sets of plots, we fix some parameters, and draw the implied volatility approximations~$\sig^{(n,m)}$ 
with $m=7$ for the Merton model and $m=8$ for the Variance-Gamma one, and $n\in\{1,2,3\}$. 
We also plot the SVI smoothing of $\sig^{(3,m)}$ as well as the true implied volatility.
The true option price is computed by a quadrature of the inverse Fourier transform representation~\eqref{eq:InvFourierPrice},
and the true implied volatility is computed by numerical inversion of the Black-Scholes formula 
(we use a simple Newton-Raphson algorithm).
We also plot the total errors between each approximation (and $\sigma^{(3,\cdot)}$ with SVI smoothing) and the true implied volatility.
As discussed above, the implied volatility approximation oscillates around the true implied volatility $\sig$.  
However, the relative error corresponding to $\sig^{(3,m)}$ is less than one percent for nearly all log-moneyness to maturity ratios (LMMRs) satisfying 
$|\zeta-x|/t<1.4$ for both models, 
which is well within the implied volatility bid-ask spread of S\&P $500$ options.
Furthermore, the relative error of $\sig^{(3,m)}$ with SVI smoothing is about one half percent for all $|\zeta-x|/t<1.0$.
As a no-arbitrage consistency check, we also plot the density corresponding to the SVI fit.
The parameters for each model are as follows:
\begin{align}
\text{Merton model: } & & 
\sig_0=0.55, a=0.25, m=-0.15, s=0.3, \alpha=1.5, t=1, x=0, \label{eq:paramMerton}\\
\text{Variance Gamma model: } &  & 
\sig_0=0.55, a=0, M=7, G=6, \alpha=4.5, t=1, x=0.\label{eq:paramVG}
\end{align}

%%%%%%%%%%%%%%%%%%%%%%%%%%%%%%%%%%%%%%%%%%%%%%%%%%%%
% 		Example: Heston
%%%%%%%%%%%%%%%%%%%%%%%%%%%%%%%%%%%%%%%%%%%%%%%%%%%%

\subsection{The Heston model}\label{sec:Heston}
In the Heston model~\cite{heston1993}, the risk-neutral dynamics of $(X,Y)$ are given by
\begin{align*}
\dd X_t &= -\frac{1}{2} Y_t \dd t + \sqrt{Y_t} \dd W_t , &
\dd Y_t &= \kappa (\theta - Y_t) \dd t + \del \sqrt{Y_t} \dd B_t , &
\dd \< W, B \>_t	&=	\rho \, \dd t ,
\end{align*}
with $(X_0,Y_0) = (x,y)\in\mathbb{R}\times (0,\infty)$, $\kappa, \theta, \delta>0$, and where $W$ and $B$ are two standard Brownian motions 
with correlation~$\rho\in [-1,1]$.
Its characteristic function reads $\phi(t,\lam,y) = C(t,\lam) + y D(t,\lam)$, where
\begin{equation*}
\left.
\begin{array}{rll}
C(t,\lambda) & := \displaystyle \frac{\kappa\theta}{\delta^2} \left((\kappa-\I \rho\delta\lambda - d(\lambda)) t
 -2\log\left[\frac{1-\gamma(\lambda) \E^{-d(\lambda)t}}{1-\gamma(\lambda)}\right]\right),
& 
d(\lambda) := \displaystyle \left(\delta^2 (\lambda^2 + \I \lambda) + 
(\kappa - \rho \I \lambda \delta)^2\right)^{1/2},\\
D(t,\lambda) & := \displaystyle \frac{\kappa - \I\rho \delta \lambda - d(\lambda)}{\delta^2}
 \frac{1-\E^{-d(\lambda)t}}{1-\gamma(\lambda) \E^{-d(\lambda)t}},
& 
\gamma(\lambda) := \displaystyle \frac{\kappa - \I\rho \delta \lambda - d(\lambda)}{\kappa - \I\rho \delta \lambda + d(\lambda)}. 
\end{array}
\right.
\end{equation*}
Unlike the exponential L\'evy setting, there is no simple general formula for the coefficients $a_n(t,\sigma_0)$ 
$(n \geq 2)$. 
However, from~\eqref{eq:Coefak}, one can compute
\begin{align}
a_2(t;\sig_0) = & \frac{\E^{2 \kappa t}}{16 \kappa^3} 
\Big(
 4 \E^{\kappa t} \left[2(\theta-y) \kappa^2+2(y+y \kappa t-\theta(2+\kappa t))\kappa\rho \del
 + (\theta + (\theta - y) \kappa t) \del^2\right] - (2 y-\theta ) \del^2
\Big) \\
 & + \frac{1}{16 \kappa^3}
\Big(
8\kappa^2(y+(\kappa t -1)\theta) - 8(y+\theta(\kappa t - 2))\kappa\rho\del - ((5-2\kappa t)\theta-2y) \del^2
\Big)
 - \frac{\sig_0^2 t}{2}.
\end{align}
Higher order terms $(3 \leq n \leq 6)$ are easily computed using any mathematical software, 
and are omitted here for clarity.
In Figure~\ref{fig:Heston}, we plot the function~$\zeta\mapsto\sig_n^{(m)}(\zeta)$ with $m=6$ and $n\in\{1,2,3\}$, 
a calibrated SVI to $\sig_3^{(6)}$ and the true implied volatility (computed exactly as for the L\'evy models above).
We also plot the relative errors between each approximation (and the SVI smoothing of $\sigma^{(3,6)}$) and the true implied volatility.
Again the approximation $\sig^{(n,m)}$ oscillates around the true implied volatility, but
the relative error of $\sig^{(3,6)}$ is less than two percent for nearly all LMMRs satisfying $|\zeta-x|/t<2.0$, 
and that of $\sig^{(3,6)}$ with SVI smoothing is roughly one percent for all 
$|\zeta-x|/t<2.0$.
As before, we also plot the density corresponding to the calibrated SVI parameterisation as a no-arbitrage consistency check.
We use the following set of parameters:
$\sig_0=0.95$, $\kappa=1$, $\theta=0.3$, $\del=0.7$, $\rho=-0.3$, $t=1$, $x=0$, $y=0.5$.

%%%%%%%%%%%%%%%%%%%%%%%%%%%%%%%%%%%%%%%%%%%%%%%%%%%%
% 		Example: Model-free calibration
%%%%%%%%%%%%%%%%%%%%%%%%%%%%%%%%%%%%%%%%%%%%%%%%%%%%

\subsection{Model-free calibration}\label{sec:Calibration}
As noted previously, the model-specific dependence of the approximate implied volatility expansion $\sig^{(n,m)}$ is entirely captured by the coefficients $a_i(t,\sig_0)$ $(2 \leq i \leq m)$.  
This simple structure allows for a model-free calibration of the implied volatility surface.  
Assume one observes implied volatilities for maturities~$(t_i)_{i=1,\ldots, n_T}$ and~$(k_j)_{j=1,\ldots,n_K}$,
where~$n_T$ and~$n_K$ are two integers. 
We shall assume for simplicity that the number of available strikes is the same for each maturity.
We suggest the following procedure:
\begin{itemize}
\item[(i)] Let $\sig_{i,j}:=\sig(t_i,k_j)$ be the quoted implied volatility for an option with maturity $t_i$ and $\log$ strike $k_j$.
\item[(ii)] Let $\sig_{i,j}^{(n,m)}:=\sig^{(n,m)}(t_i,k_j;\sig_0)$ be the approximate implied volatility for an option 
with maturity $t_i$ and $\log$ strike $k_j$ computed using the approximation~\eqref{eq:sig.n.m}.
\item[(iii)] At each maturity $t_i$, leave $\sig_0$ and $a_q(t_i;\sig_0)$ $(2\leq q \leq m)$ as free parameters.   Fit $\sig^{(n,m)}(t_i,\cdot)$ to the market's $t_i$-maturity implied volatility smile $\sig(t_i,\cdot)$ by minimising
$\sum_{j=1}^{n_K} \left|\sig^{(n,m)}(t_i,k_j) - \sig_{ij}^{(n,m)}\right|^2$.
\item[(iv)] As an initial guess, use the largest quoted implied volatility at each maturity for $\sigma_0$,
and $a_q(t_i,\sig_0)=0$.
\end{itemize}
\begin{remark}
With $n=3$ and $m=8$, step~(iii) is instantaneous using Mathematica's {\tt FitTo} or Matlab's {\tt lsqnonlin} for instance.
\end{remark}
We test this procedure on SPX index options from January 4, 2010 with $n=3$ and $m=8$.  
The results for three separate maturities ($t=0.033$, $t=0.70$, $t=1.45$ years) 
are given on Figure~\ref{fig:SPX-FIT}.
The calibrated parameters are ($a_i$ is a shorthand notation for $a_i(t;\sig_0)$):
\begin{center}
  \begin{tabular}{ l | c c c c c c c c }
 & $\sigma_0$ &  $a_2$ &  $a_3$ & $a_4$ & $a_5$ & $a_6$ & $a_7$ & $a_8$\\
\hline
$t=$0.033 & 0.382 &  -1.64E-3 & -1.00E-6 & 1.64E-6 & 1.89E-8 & 8.62E-10 & 5.22E-12 & 1.40E-13\\
$t=$0.70 & 0.659 & -1.31E-1 & -5.00E-3 & 2.40E-3 & 4.24E-4 & 6.96E-5 & 2.88E-6 & 4.83E-7\\
$t=$1.45 & 0.436 & -8.35E-2 & -1.22E-2 & 1.68E-3 & -6.26E-5 & 3.42E-5 & -1.77E-6 & 3.96E-7
  \end{tabular}
\end{center}

\begin{remark}
If the stock price is an exponential L\'evy model, then~\eqref{eq:Levy.a} implies that
$\frac{1}{2} \sig_0^2 + \frac{1}{t_i} a_2(t_i,\sig_0)=\frac{1}{2}(a^2 + I_2)$ 
and $\frac{1}{t_i}a_q(t_i) = \frac{1}{q!}I_q$ $(3\leq q \leq m)$ should be constant.  
If this is not so, then exponential L\'evy models are probably not the best dynamics to describe the underlying.
\end{remark}

\begin{remark}
Our whole methodology is based on approximating the characteristic function of a process by a truncated version
of its expansion with respect to some small parameter.
In essence, this truncation tends to ignore the tail behaviour (high-order terms in the expansion) of the process,
and hence, even though the resulting volatility expansion is accurate around the money, 
there is no reason why it should be so in the tails.
The latter, however, are usually not observable in practice, so that this should be of lesser concern for practical implementation.
This in particular means that, should one plot the densities corresponding to the fit in Figure~\ref{fig:SPX-FIT},
the latter may become negative in the tails (hence allowing for arbitrage opportunities).
In our calibration example (Figure~\ref{fig:SPX-FIT}), the density does remain non-negative though.
If however it was to become negative, one could perform an SVI fit, as explained in Section~\ref{sec:SVI}, or, even better,
use the fully no-arbitrage SVI version developed in~\cite{gatheral2012arbitrage}.
\end{remark}

%%%%%%%%%%%%%%%%%%%%%%%%%%%%%%%%%%%%%%%%%%%%%%%%%%%%
%		SECTION: Thanks
%%%%%%%%%%%%%%%%%%%%%%%%%%%%%%%%%%%%%%%%%%%%%%%%%%%%

\subsection*{Acknowledgments}
M.~Lorig acknowledges financial support from Imperial College London, and both authors are grateful to Claude Martini and Zeliade Systems for their useful comments and for indicating to us the availability of CBOE data.  The authors would also like to thank an anonymous referee, whose suggestions improved both the mathematical quality and readability of this manuscript.

%%%%%%%%%%%%%%%%%%%%%%%%%%%%%%%%%%%%%%%%%%%%%%%%%%%%
%
%			Appendix
%
%%%%%%%%%%%%%%%%%%%%%%%%%%%%%%%%%%%%%%%%%%%%%%%%%%%%

\appendix

\section{Proof of Theorem \ref{thm:main}}
\label{sec:proof}
From~\eqref{eq:sig.k.m} we observe that $\sig_k^{(m)}$ involves both
$\d_\sig^n u^\BS/\d_\sig u^\BS$ and $u_k^{(m)}/\d_\sig u^\BS$.  
Theorem~\ref{thm:main} will follow directly from Definition~\ref{def:sig.nm}, 
and Propositions~\ref{prop1} and~\ref{prop2},
both of which provide explicit expressions for these quantities.
In all the statements and results below, we shall consider 
$(t, x, \zeta, \sigma)\in\mathbb{R}_+\times\mathbb{R}\times\mathbb{R}\times (0,\infty)$, 
and define $y := \frac{1}{\sig\sqrt{2t}}\(x-\zeta-\frac{1}{2}\sig^2 t\)$.
Furthermore, we recall that, for any $n\in\mathbb{N}$, 
$H_n$ denotes the $n$-th Hermite polynomial from Theorem~\ref{thm:main}.
We first start with the following lemma.
\begin{lemma}
\label{lemma:ratio}
For any integers $m\geq 0$, $n \geq 2$ we have
\begin{align}\label{eq:ratio}
\frac{\d_x^m (\d_x^n - \d_x ) u^{\BS}}{(\d_x^2 - \d_x ) u^{\BS}}(t,x,\sig) 
 = \sum_{i=2}^n \left(-\frac{1}{\sig\sqrt{2t}}\right)^{m+i-2} H_{m+i-2}(y).
\end{align}
%holds for any $x\in\mathbb{R}$, where $H_n$ is the $n$-th Hermite polynomial, recalled in Theorem~\ref{thm:main}.
%where $H_n(y) \equiv (-1)^n \E^{y^2} \d_y^n \E^{-y^2}$ is the $n$-th Hermite polynomial.
\end{lemma}
\begin{proof}
From the Black-Scholes call price formula \eqref{eq:uBS.2}
we clearly obtain
$(\d_x^2 - \d_x ) u^{\BS}(t,x,\sig)
 = \frac{1}{\sig \sqrt{t}} \exp \( x - \frac{1}{2}d_+^2(x)\)$,
where 
$d_\pm(x) := \frac{1}{\sig \sqrt{t}} \( x - \zeta \pm \frac{1}{2}\sig^2 t \)$.
Now, for any integers $m \geq 0$ and $n \geq 2$, we have
\begin{align}
\d_x^m (\d_x^n - \d_x ) u^{\BS}(t,x,\sig)
	&=	\d_x^m \sum_{i=2}^n (\d_x^i - \d_x^{i-1} ) u^{\BS}(t,x,\sig)
	=	\sum_{i=2}^n \d_x^{m+i-2} (\d_x^2 - \d_x ) u^{\BS}(t,x,\sig) \\
	&=	\frac{1}{\sig \sqrt{t}} \sum_{i=2}^n \d_x^{m+i-2} \exp \( x - \frac{d_+^2(x)}{2} \) .
 \label{eq:dm.dn-d}
\end{align}
The lemma then follows from the identity $x-\frac{1}{2}d_+^2(x) = - y^2 + \zeta$  and from
\begin{align}
\frac{\d_x^m (\d_x^n - \d_x ) u^{\BS}(t,x,\sig)}{(\d_x^2 - \d_x ) u^{\BS}(t,x,\sig)} 
	&= \frac{\sum_{i=2}^n \d_x^{m+i-2} \exp \( x - \frac{d_+^2(x)}{2} \)}{\exp \( x - \frac{d_+^2(x)}{2} \)} \\
	&=	\sum_{i=2}^n \E^{y^2} \(\frac{1}{\sig\sqrt{2t}}\)^{m+i-2} \d_y^{m+i-2}\left(\E^{-y^2}\right)
	=	\sum_{i=2}^n \(-\frac{1}{\sig\sqrt{2t}}\)^{m+i-2} H_{m+i-2}(y).
\end{align}
\end{proof}
\begin{proposition}\label{prop1}
The identity 
\begin{align}
\frac{\d_\sig^n u^{\BS}}{\d_\sig u^{\BS}}(t,x,\sigma)
	&=	\sum_{q=0}^{\left\lfloor n/2 \right\rfloor}\sum_{p=0}^{n-q-1}
			\binom{n-q-1}{p}c_{n,n-2q}\sig^{n-2q-1} t^{n-q-1} 
			\left(\sig\sqrt{2 t}\right)^{1-p-n+q} H_{p+n-q-1}(y) , \label{eq:2}
\end{align}
holds, where the coefficients $(c_{n,n-2k})$ are defined recursively by
$c_{n,n}=1$ and $c_{n,n-2q} =(n-2q+1) c_{n-1,n-2q+1} + c_{n-1,n-2q-1}$
for $q \in \{ 1, 2, \ldots, \left\lfloor n/2 \right\rfloor \}$.
\end{proposition}
\begin{proof}
For any integers~$j\geq 2$ and~$k\geq 2$ we can write
\begin{align}
(\d_x^j - \d_x)^k
	&=	(\d_x^j - \d_x)^{k-1} (\d_x^j - \d_x) 
	=	\sum_{l=0}^{k - 1} \binom{k-1}{l} (\d_x^j)^l (-\d_x)^{k-1-l} (\d_x^j - \d_x) \\
	&=	\sum_{l=0}^{k - 1} \binom{k-1}{l} (-1)^{k-1-l}\d_x^{n (j-1) + k-1} (\d_x^j - \d_x) .
\end{align}
Combining this with Lemma~\ref{lemma:ratio} and the identity
$\d_\sig u^\BS = t \sig (\d_x^2 - \d_x) u^\BS$, we obtain
\begin{align}
\frac{(\d_x^j - \d_x)^k u^\BS}{\d_\sig u^\BS}
	&=	\frac{1}{t \sig}\sum_{l=0}^{k - 1} \binom{k-1}{l} (-1)^{k-1-l}
			\frac{\d_x^{l (j-1) + k-1} (\d_x^j - \d_x) u^\BS}{(\d_x^2-\d_x) u^\BS} \\
	&=	\frac{1}{t \sig}\sum_{l=0}^{k - 1} \sum_{i=2}^j \binom{k-1}{l} (-1)^{k-1-l}
			\left(-\frac{1}{\sig\sqrt{2t}}\right)^{l (j-1) + k-1+i-2} H_{l (j-1) + k-1+i-2}(y) . \label{eq:result}
\end{align}
Define the operator $\Lc := t (\d_x^2 - \d_x)$ 
(so that $\d_\sig u^\BS = \sig \Lc u^\BS$); 
for any $n\in\mathbb{N}$, the identity
$$
\d_\sig^n u^\BS 
	= \sum_{q=0}^{\left\lfloor n/2 \right\rfloor} c_{n,n-2q}\sig^{n-2q}\Lc^{n-q} u^\BS, %\label{eq:claim}
$$
follows from a simple recursion, with the coefficients $(c_{n, n-2k})$ defined
as in the proposition.
%where $c_{n,n}=1$ and $c_{n,n-2q} =(n-2q+1) c_{n-1,n-2q+1} + c_{n-1,n-2q-1}$
%for any integer $q \in \{ 1, 2, \cdots , \left\lfloor n/2 \right\rfloor\}$.
Therefore, using~\eqref{eq:result} with $j=2$ and $k=n-q$,
%Using~\eqref{eq:claim}, we obtain
\begin{align}
\frac{\d_\sig^n u^\BS}{\d_\sig u^\BS}
	&=	\sum_{q=0}^{\left\lfloor n/2 \right\rfloor} c_{n,n-2q}\sig^{n-2q}
			\frac{\Lc^{n-q} u^\BS}{\d_\sig u^\BS} 
	=	\sum_{q=0}^{\left\lfloor n/2 \right\rfloor} c_{n,n-2q}\sig^{n-2q} t^{n-q}
			\frac{(\d_x^2-\d_x)^{n-q} u^\BS}{\d_\sig u^\BS} \\
	&=	\sum_{q=0}^{\left\lfloor n/2 \right\rfloor}\sum_{p=0}^{n-q-1}
			c_{n,n-2q}\sig^{n-2q-1} t^{n-q-1} \binom{n-q-1}{p} (-1)^{n-q-1-p}
			\left(-\frac{1}{\sig\sqrt{2t}}\right)^{p+n-q-1} H_{p+n-q-1}(y).
\end{align}
\end{proof}
\begin{proposition}
\label{prop2}
%Recall that $H_n$ is the $n$-th Hermite polynomial, recalled in Theorem \ref{thm:main}
The following equality holds:
%For any $(t,x,k,\sig_0)$, we have
%and let $u^\BS=u^\BS(t,x,\sig_0)$ and $u_n^{(m)}=u_n^{(m)}(t,x,\sig_0)$.  Then
\begin{align*}
\frac{u_n^{(m)}}{\d_\sig u^\BS}(t,x,\sig)
	&=	\frac{1}{n! t \sig}  \sum_{k=2 \cdot n}^{n \cdot m} 
			\sum_{\substack{j_1+\cdots+j_n = k \\ 2 \leq j_1,\cdots,j_n \leq m}} 
			\( \prod_{i=1}^n a_{j_i} \) \sum_{k_1=2}^{j_1} \cdots \sum_{k_n=2}^{j_n} 
			\binom{n-1}{m}(-1)^{n-1-m} \\ & \qquad \times
			\(-\frac{1}{\sig\sqrt{2t}}\)^{-n-1+m+\sum_{j=1}^n {k_j}}
			H_{-n-1+m+\sum_{j=1}^n {k_j}}(y),
\end{align*}
\end{proposition}
\begin{proof}
An expansion (and slight reorganisation of the terms) of~\eqref{eq:u.n.m} yields
\begin{equation}\label{eq:A}
u_n^{(m)} =
\frac{1}{n!} \( \sum_{k=2 n}^{n m} 
\sum_{\substack{j_1+\cdots+j_n = k \\ 2 \leq j_1,\cdots,j_n \leq m}} 
\prod_{i=1}^n a_{j_i} (\d_x^{j_i}-\d_x) \) u^\BS.
\end{equation}
Furthermore, we have
\begin{align}
\prod_{i=1}^n a_{j_i} (\d_x^{j_i}-\d_x) 
	&=	\prod_{i=1}^n a_{j_i} \sum_{k_i=2}^{j_i}(\d_x^{k_i}-\d_x^{k_i-1}) 
	=		\prod_{i=1}^n a_{j_i} \sum_{k_i=2}^{j_i} \d_x^{k_i-2} (\d_x^2-\d_x) \\
	&=	\( \prod_{i=1}^n a_{j_i}\) \sum_{k_1=2}^{j_1} \cdots \sum_{k_n=2}^{j_n} 
			\d_x^{k_1-2} \cdots \d_x^{k_n-2} (\d_x^2-\d_x)^{n-1} (\d_x^2-\d_x) \\
	&=	\( \prod_{i=1}^n a_{j_i}\) \sum_{k_1=2}^{j_1} \cdots \sum_{k_n=2}^{j_n} 
			\d_x^{k_1-2} \cdots \d_x^{k_n-2} \binom{n-1}{m}\d_x^{n-1+m}(-1)^{n-1-m} (\d_x^2-\d_x) \\
	&=	\( \prod_{i=1}^n a_{j_i}\) \sum_{k_1=2}^{j_1} \cdots \sum_{k_n=2}^{j_n} 
			\binom{n-1}{m}(-1)^{n-1-m} \d_x^{-n-1+m+\sum_{j=1}^n {k_j}}(\d_x^2-\d_x) . \label{eq:B}
\end{align}
Using Lemma~\ref{lemma:ratio} and 
the equality $\d_\sig u^\BS (t,x,\sig) = t \sig ( \d_x^2 - \d_x ) u^\BS (t,x,\sig)$, we observe
$$
\frac{\d_x^{-n-1+m+\sum_{j=1}^n {k_j}}(\d_x^2-\d_x)u^\BS}{\d_\sig u^\BS}(t,x,\sig)
 = \frac{1}{t \sig} \(-\frac{1}{\sig_0\sqrt{2t}}\)^{-n-1+m+\sum_{j=1}^n {k_j}} H_{-n-1+m+\sum_{j=1}^n {k_j}}(y).
$$
Combining this with~\eqref{eq:A} and~\eqref{eq:B} concludes the proof of the proposition.
\end{proof}

%%%%%%%%%%%%%%%%%%%%%%%%%%%%%%%%%%%%%%%%%%%%%%%%%%%%
%
%			Bibliography
%
%%%%%%%%%%%%%%%%%%%%%%%%%%%%%%%%%%%%%%%%%%%%%%%%%%%%

%\begin{small}
\bibliographystyle{chicago}
\bibliography{BibTeX-Master}	
%\end{small}

%%%%%%%%%%%%%%%%%%%%%%%%%%%%%%%%%%%%%%%%%%%%%%%%%%%%
%
%			Figures
%
%%%%%%%%%%%%%%%%%%%%%%%%%%%%%%%%%%%%%%%%%%%%%%%%%%%%

%%%%%%%%%%%%%%%%%%%%%%%%%%%%%%%%%%%%%%%%%%%%%%%%%%%%
%			 Merton
%%%%%%%%%%%%%%%%%%%%%%%%%%%%%%%%%%%%%%%%%%%%%%%%%%%%

\begin{figure}[ht]\hspace{5pt}
\centering
\subfigure[Implied volatilities]{\includegraphics[scale=0.5]{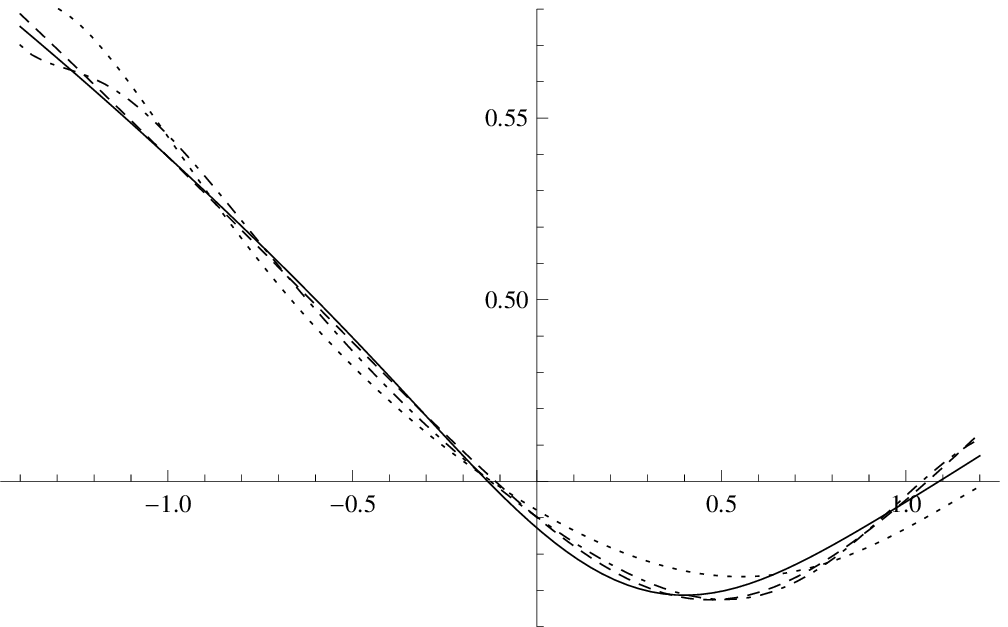}}
\hspace{5pt}
\subfigure[Implied volatility errors]{\includegraphics[scale=0.5]{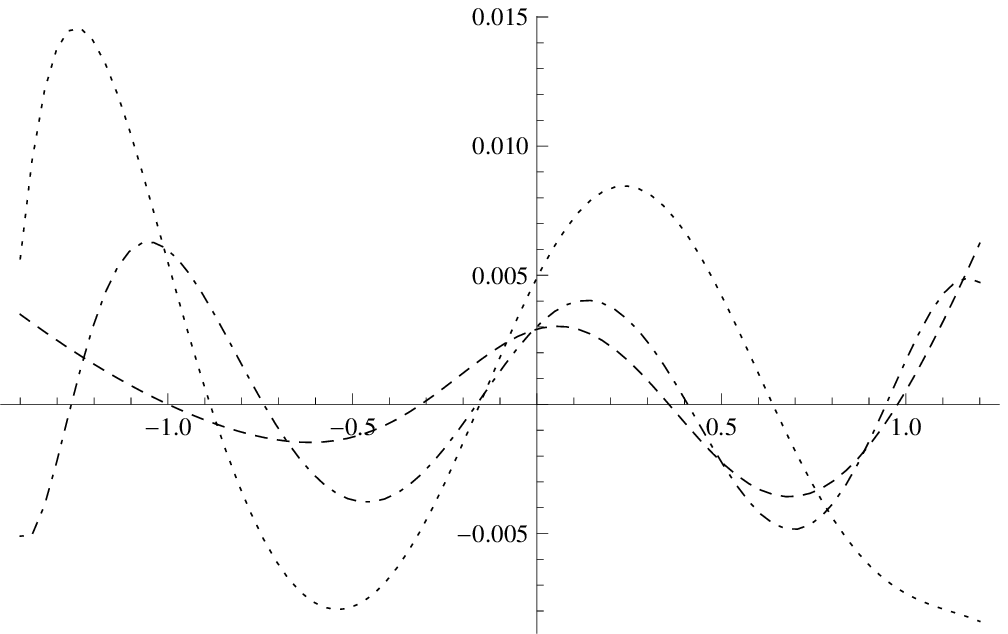}}\\
\subfigure[Density]{\includegraphics[scale=0.5]{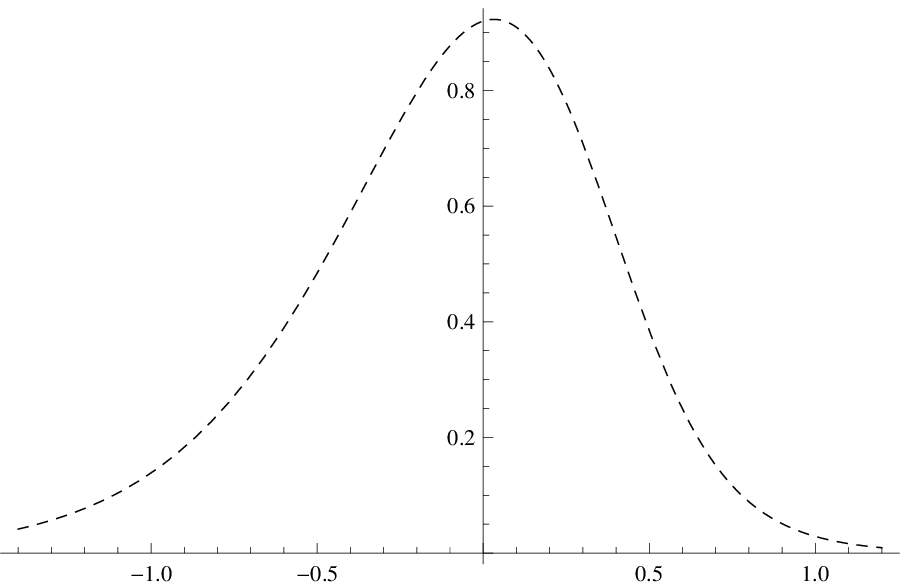}}
\hspace{5pt}
\subfigure[Tails of the density]{\includegraphics[scale=0.5]{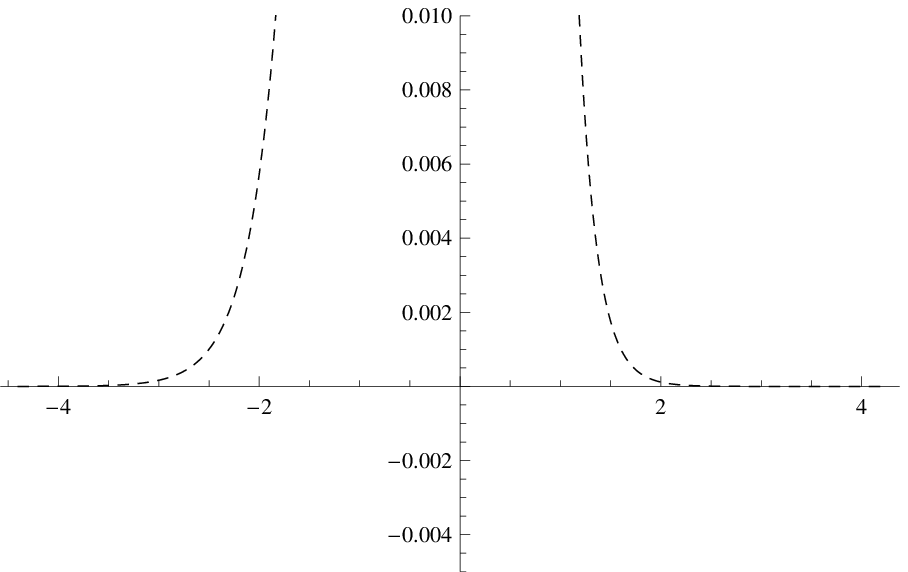}}
\caption{Numerics for the Merton model detailed in Section~\ref{sec:Levy}. 
The top graphs correspond to the true smile (solid line), the approximations $\sigma^{(2,7)}$ (dots) 
and $\sigma^{(3,7)}$ (dots-dashed), and the SVI smoothing of $\sigma^{(3,7)}$ (dashed).
The plots below are the density (and its tails) of the SVI smoothing.}
\label{fig:Merton}
\end{figure}

%%%%%%%%%%%%%%%%%%%%%%%%%%%%%%%%%%%%%%%%%%%%%%%%%%%%
%			 Variance Gamma
%%%%%%%%%%%%%%%%%%%%%%%%%%%%%%%%%%%%%%%%%%%%%%%%%%%%

\begin{figure}[ht]\hspace{5pt}
\centering
\subfigure[Implied volatilities]{\includegraphics[scale=0.5]{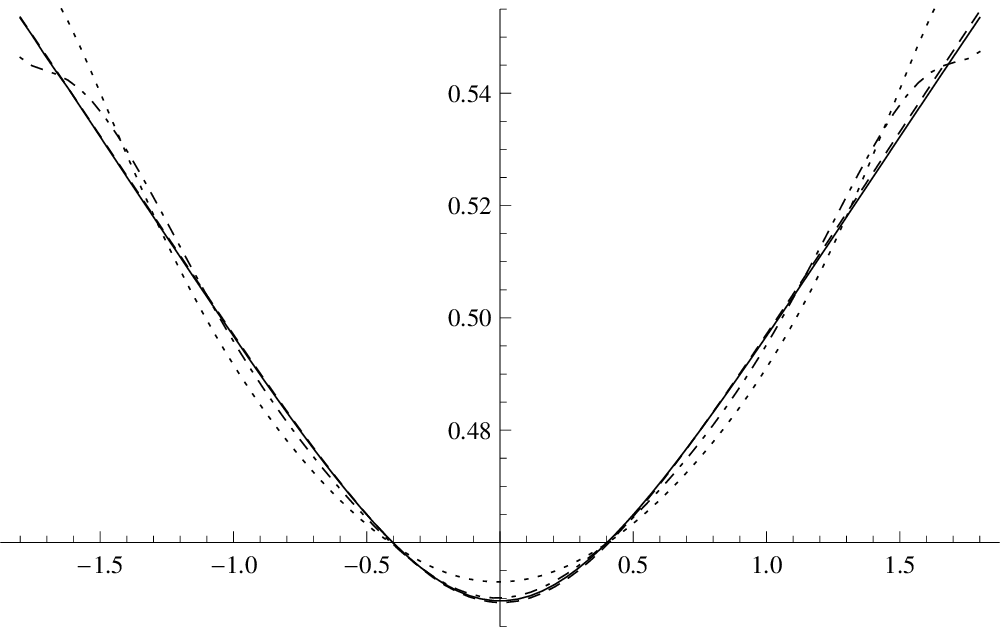}}
\hspace{5pt}
\subfigure[Implied volatility errors]{\includegraphics[scale=0.5]{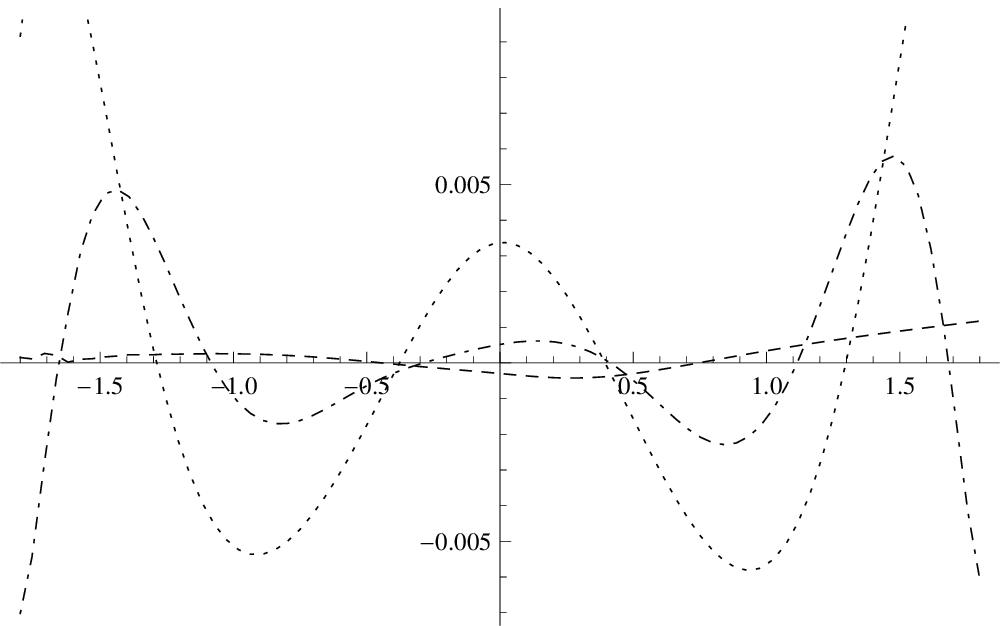}}\\
\subfigure[Density]{\includegraphics[scale=0.5]{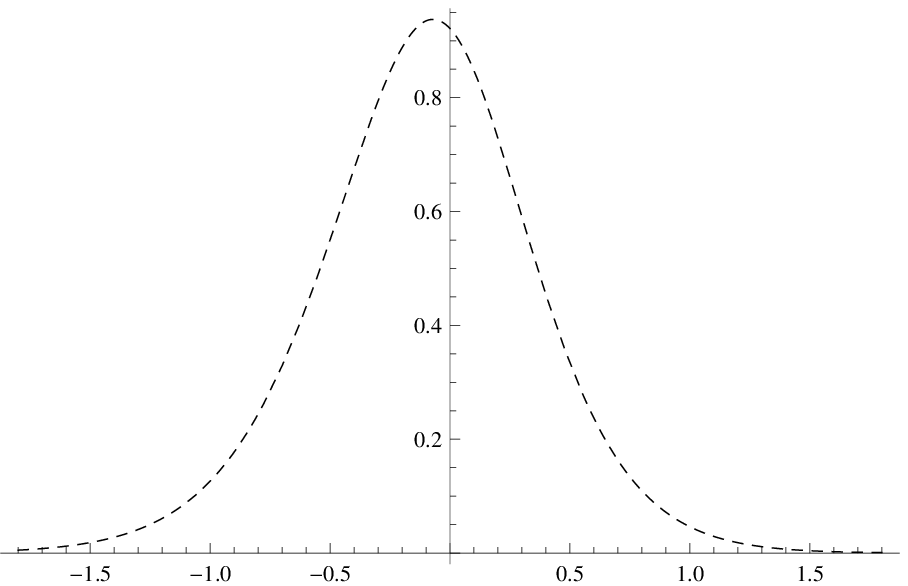}}
\hspace{5pt}
\subfigure[Tails of the density]{\includegraphics[scale=0.5]{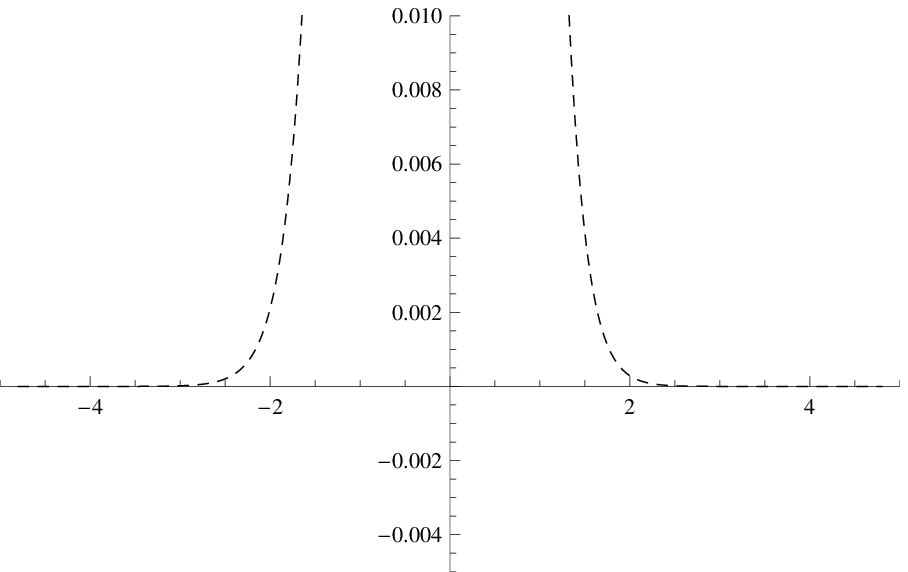}}
\caption{Numerics for the Variance Gamma model detailed in Section~\ref{sec:Levy}. 
The top graphs correspond to the true implied volatility (solid line), the approximations
$\sigma^{(2,8)}$ (dots) and $\sigma^{(3,8)}$ (dots-dashed), 
and the SVI smoothing of $\sigma^{(3,8)}$ (dashed).
The plots below are the density (and its tails) of the SVI smoothing.}
\label{fig:VG}
\end{figure}

%%%%%%%%%%%%%%%%%%%%%%%%%%%%%%%%%%%%%%%%%%%%%%%%%%%%
%			 Heston
%%%%%%%%%%%%%%%%%%%%%%%%%%%%%%%%%%%%%%%%%%%%%%%%%%%%

\begin{figure}[ht]\hspace{5pt}
\centering
\subfigure[Implied volatilities]{\includegraphics[scale=0.5]{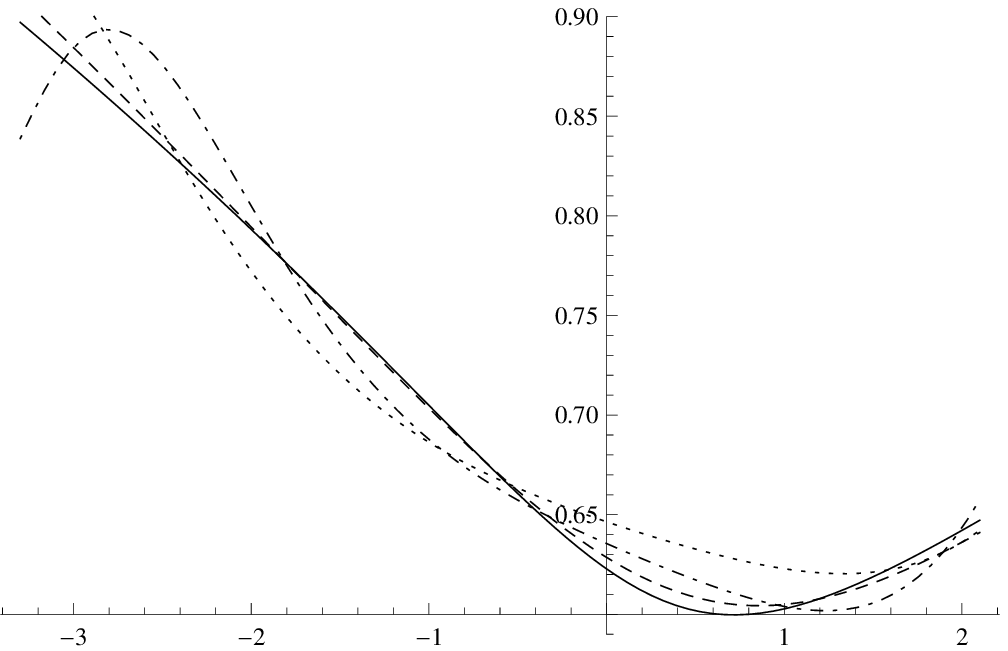}}
\hspace{5pt}
\subfigure[Implied volatility errors]{\includegraphics[scale=0.5]{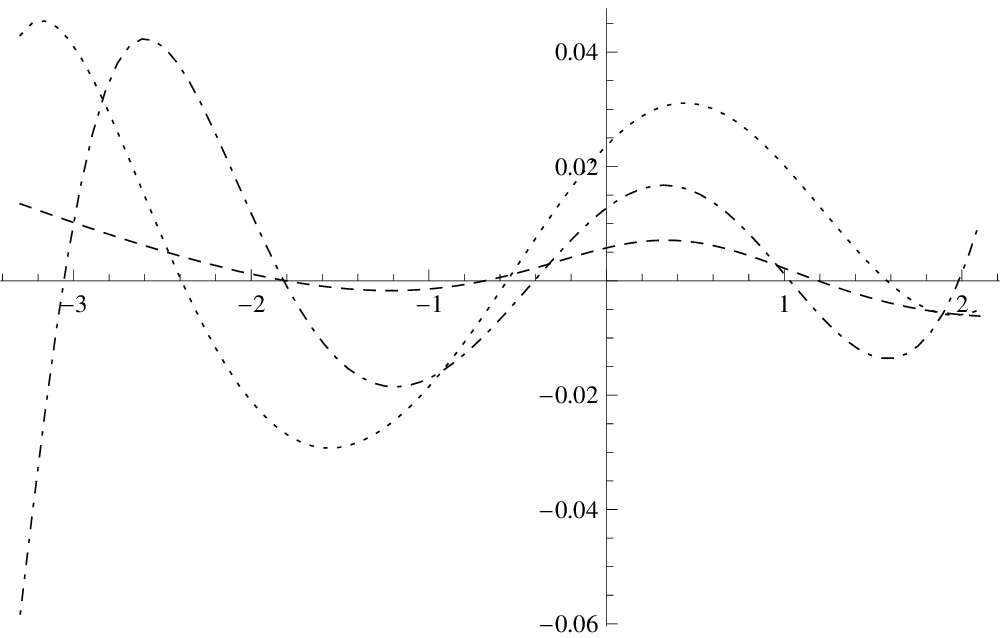}}\\
\subfigure[Density]{\includegraphics[scale=0.5]{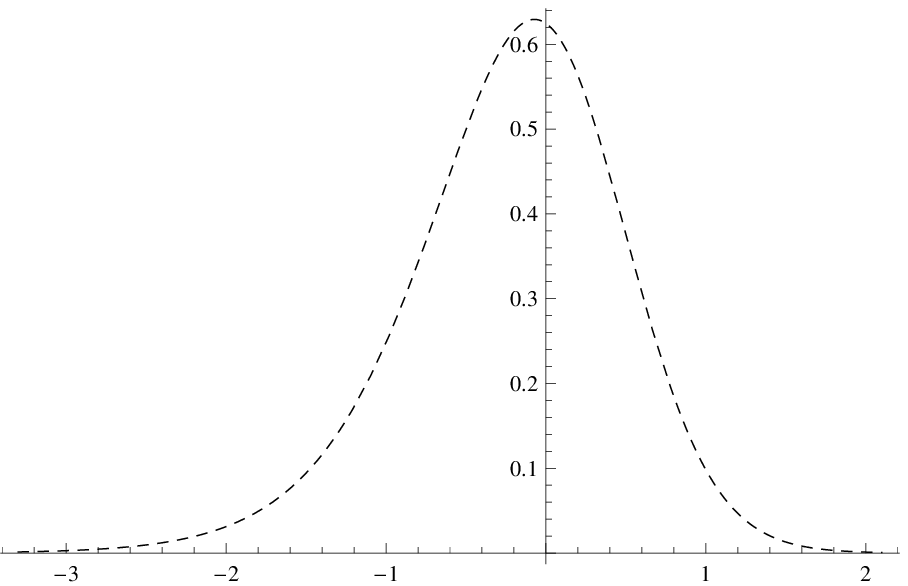}}
\hspace{5pt}
\subfigure[Tails of the density]{\includegraphics[scale=0.5]{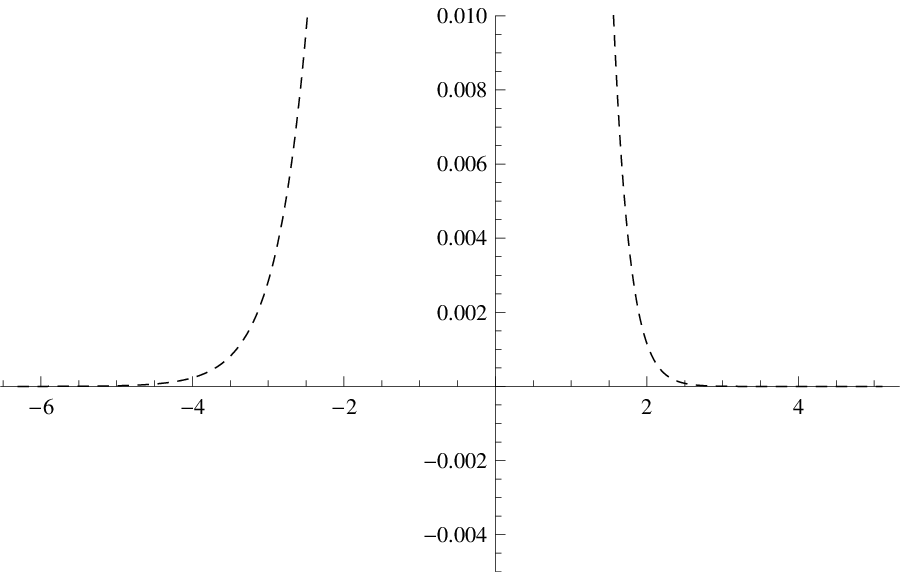}}
\caption{
Numerics for the Heston model detailed in Section~\ref{sec:Heston}. 
The top graphs correspond to the true implied volatility (solid line), 
the approximations $\sigma^{(2,6)}$ (dots) and $\sigma^{(3,6)}$ (dots-dashed), 
and the SVI smoothing of $\sigma^{(3,6)}$ (dashed).
The plots below are the density (and its tails) of the SVI smoothing.}
\label{fig:Heston}
\end{figure}

%%%%%%%%%%%%%%%%%%%%%%%%%%%%%%%%%%%%%%%%%%%%%%%%%%%%
%			 Calibration
%%%%%%%%%%%%%%%%%%%%%%%%%%%%%%%%%%%%%%%%%%%%%%%%%%%%

\begin{figure}[ht]\hspace{5pt}
\centering
\subfigure[$t=0.033$ years]{\includegraphics[scale=0.5]{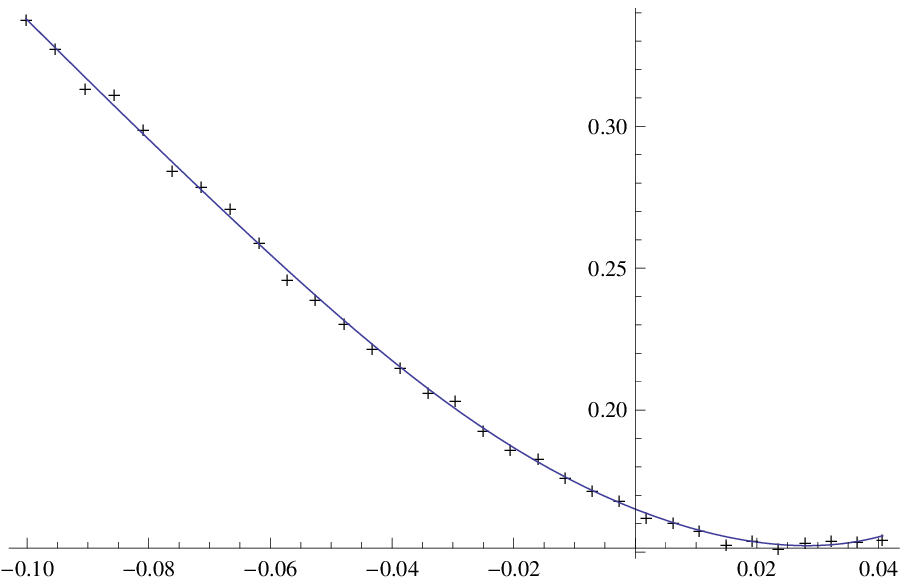}}
\hspace{5pt}
\subfigure[ $t=0.70$ years]{\includegraphics[scale=0.5]{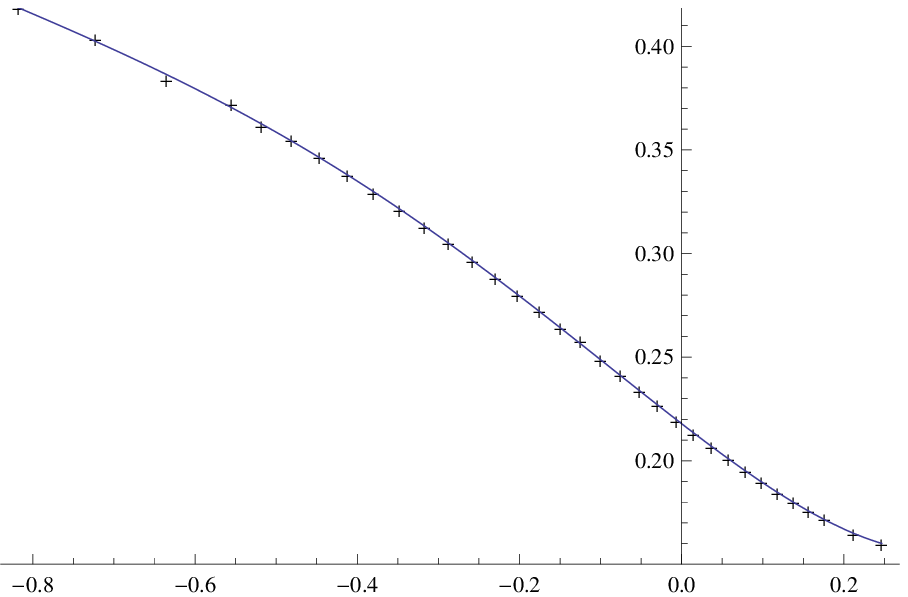}}
\hspace{5pt}
\subfigure[$t=1.45$ years]{\includegraphics[scale=0.5]{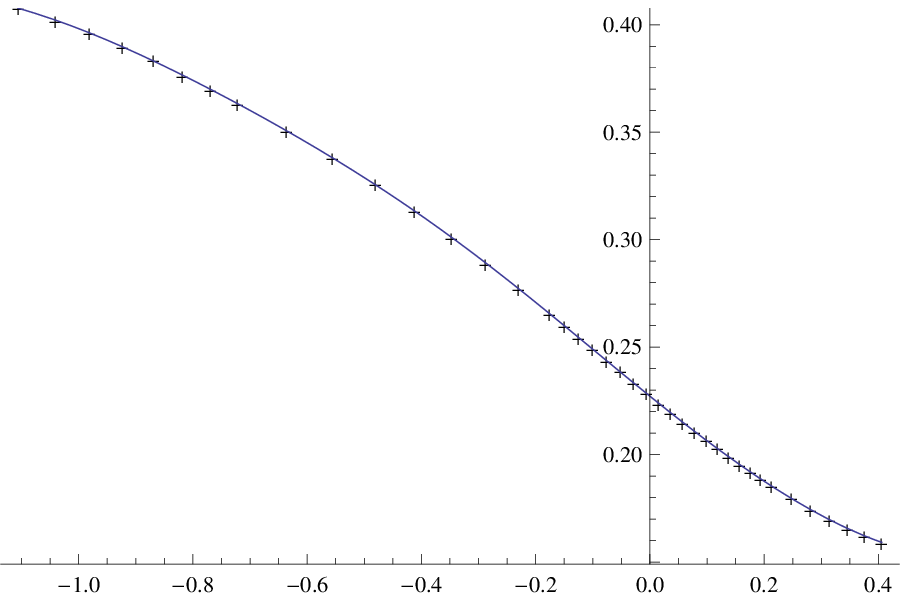}}
\hspace{5pt}
\subfigure[$t=0.033$ years]{\includegraphics[scale=0.5]{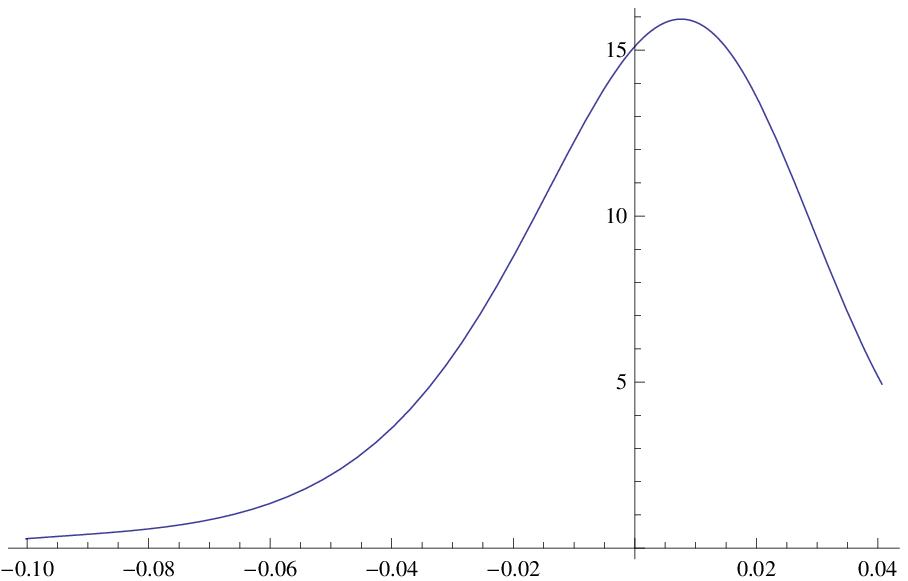}}
\hspace{5pt}
\subfigure[$t=0.70$ years]{\includegraphics[scale=0.5]{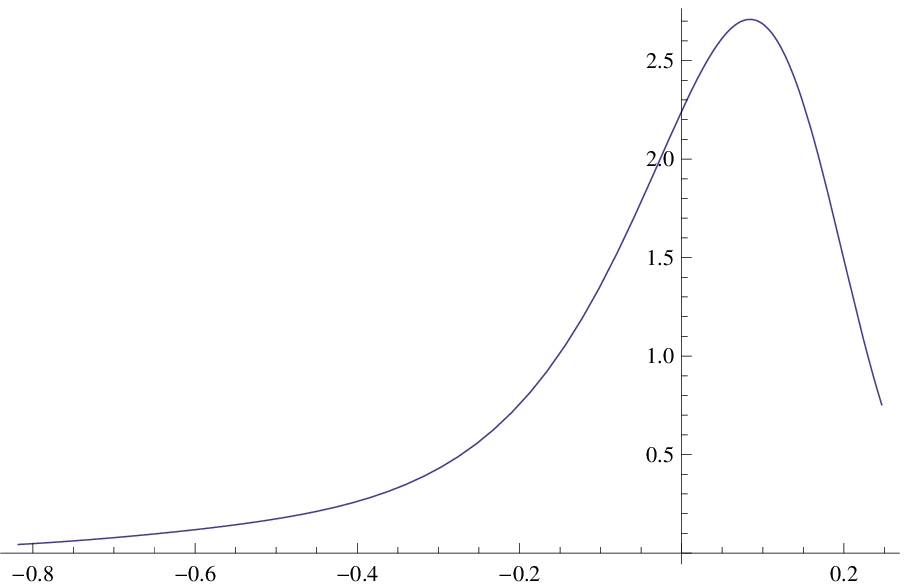}}
\hspace{5pt}
\subfigure[$t=1.45$ years]{\includegraphics[scale=0.5]{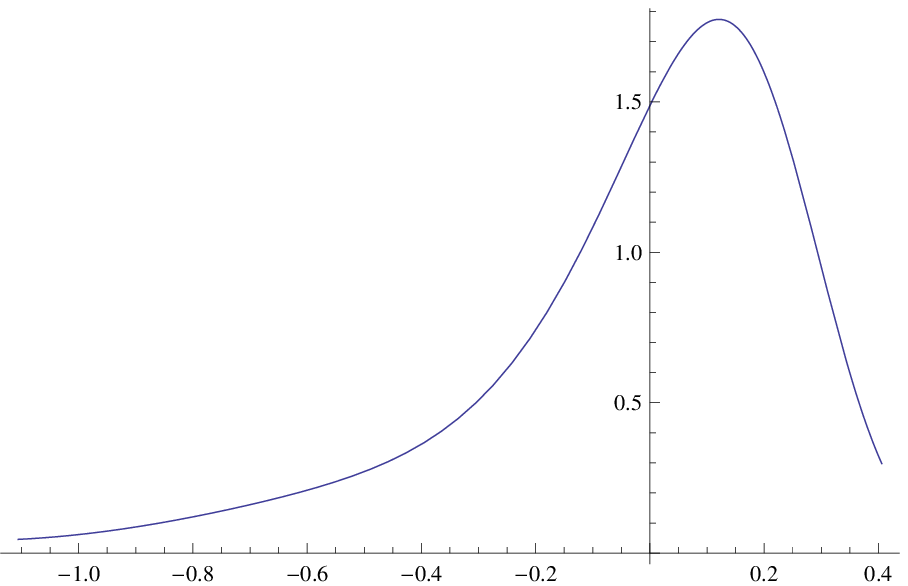}}
\caption{Model-free fit to SPX options from Jan 4, 2010 as explained in Section~\ref{sec:Calibration}.
The horizontal axis represents the log-moneyness $(\zeta-x)$.
The plots below represent the densities of the fit.}
\label{fig:SPX-FIT}
\end{figure}

\end{document}